\documentclass[final,3p,times,twocolumn]{elsarticle}
\usepackage{epstopdf}
\usepackage{graphicx}
\usepackage{amssymb}
\usepackage{amsmath}
\usepackage{amsthm}
\usepackage{caption}
\usepackage{subcaption}
\usepackage{amsfonts}
\usepackage{enumitem}
\usepackage{graphicx}
\usepackage{caption}
\usepackage{float}
\usepackage{hyperref}
\usepackage{fancyhdr}
\usepackage{lastpage}
\usepackage{wrapfig}
\usepackage{hyperref}
\usepackage{algorithm}
\usepackage{algorithmicx}
\usepackage{color}

\newtheorem{proposition}{Proposition}[section]

\usepackage{amsmath,algorithm,tabularx}
\usepackage[noend]{algpseudocode}
\makeatletter
\newcommand{\multiline}[1]{%
  \begin{tabularx}{\dimexpr0.9\linewidth-\ALG@thistlm}[t]{@{}X@{}}
    #1
  \end{tabularx}
}
\makeatother
\usepackage{lineno}

\journal{Notes}

\begin{document}

\begin{frontmatter}

\title{Generating drawdown-realistic financial price paths using path signatures}



\author{Emiel Lemahieu\footnotemark[1]\footnotemark[2]\footnotemark[3], Kris Boudt\footnotemark[1], Maarten Wyns\footnotemark[2]}

\address{Ghent University, InvestSuite}

\begin{abstract}
A novel generative machine learning approach for the simulation of sequences of financial price data with drawdowns quantifiably close to empirical data is introduced. Applications such as pricing drawdown insurance options or developing portfolio drawdown control strategies call for a host of drawdown-realistic paths. Historical scenarios may be insufficient to effectively train and backtest the strategy, while standard parametric Monte Carlo does not adequately preserve drawdowns. We advocate a non-parametric Monte Carlo approach combining a variational autoencoder generative model with a drawdown reconstruction loss function. To overcome issues of numerical complexity and non-differentiability, we approximate drawdown as a linear function of the moments of the path, known in the literature as path signatures. We prove the required regularity of drawdown function and consistency of the approximation. Furthermore, we obtain close numerical approximations using linear regression for fractional Brownian and empirical data. We argue that linear combinations of the moments of a path yield a mathematically non-trivial smoothing of the drawdown function, which gives one leeway to simulate drawdown-realistic price paths by including drawdown evaluation metrics in the learning objective. We conclude with numerical experiments on mixed equity, bond, real estate and commodity portfolios and obtain a host of drawdown-realistic paths.
\end{abstract}

\end{frontmatter}

% \linenumbers

\section{Introduction}
\label{s1_introduction}
\addtocounter{footnote}{1}
\footnotetext{Ghent University, Sint-Pietersplein 6, 9000 Gent, Belgium}
\addtocounter{footnote}{1}
\footnotetext{InvestSuite, \url{www.investsuite.com}}
\addtocounter{footnote}{1}
\footnotetext{Corresponding author at \href{mailto:emiel.lemahieu@ugent.be}{emiel.lemahieu@ugent.be}}
Market generators are generative machine learning (ML) models with the specificity of modeling financial markets such as stock price or return time series, order books, implied volatility surfaces, and more. They have gained popularity in recent years (\citet{wiese2020quant}, \citet{koshiyama2021generative}, \citet{cont2022tail}, \citet{bergeron2022variational}).  Arguably the main advantage of generative ML over existing simpler solutions such as standard Monte Carlo engines is being able to simulate realistic financial time series without having to specify a data generating process (DGP) a priori. Scenario realisticness relates to low distributional divergences between the actual and generated return distributions (such as difference in normalized moments or max mean discrepancy, entropy or Kullback-Leibler divergence, optimal transport measures like Wasserstein distance, etc.), and distances between their autocorrelation functions and tail indices. Some striking results in the latter regard have been found in \citet{wiese2020quant}, \citet{cont2022tail}, \citet{buehler2020data} and \citet{fingan2023}.

A drawdown is a price fall relative to its historical maximum. In this article, we argue that one can use a market generator as a non-parametric Monte Carlo engine in order to solve the difficult problem of generating realistic drawdown scenarios. With parametric DGP approaches, no explicit analytical link between the underlying DGP parameters and their corresponding drawdown distribution is known, unless under very restrictive assumptions. Parametric approaches to drawdown have commonly relied on Brownian assumptions as to use Lévy's theorem to obtain this analytical drawdown distribution.  Especially in applications where this measure is crucial one would need theoretical guarantees that upon convergence of the parameters of the generative model drawdowns are explicitly preserved. Examples include pricing drawdown options for max loss insurance (\citet{carr2011maximum}, \citet{atteson2022carr}) and controlling for drawdown in a portfolio optimization context (\citet{chekhlov2004portfolio}, \cite{chekhlov2005drawdown}).

The main contribution of this paper is overcoming the issues of non-differentiability and numerical complexity inherent to the drawdown function by approximating it with a linear regression on the path's  signature terms. Signatures serve as the moment generating function on the path space and are because of their universality (\citet{chevyrev2018signature}) the likely candidate for universal approximation of drawdown with a linear function. We prove the required regularity of drawdown function and the consistency of the approximation. Furthermore, we discuss the adequacy of the approximation as a function of the approximation order and for different levels of roughness and sample sizes.  We argue that by adding weighted moments of the path to the reconstruction loss of the market generator, one can obtain realistic drawdown scenarios, solving the understatement of adverse drawdown scenarios of traditional generative model architectures or common DGP assumptions.

The article is structured as follows. Section \ref{market generators section} covers the background on market generators and positions our problem setting within this literature. Section \ref{section drawdowns} discusses drawdowns, notations and its Lipschitz continuity following from its reformulation as a non-linear dynamic system. Section \ref{rough path theory section} summarizes some key elements from rough path theory and the approximation of functions on paths. We propose a linear approximation of drawdowns in the signature space. Section \ref{market generator} introduces our market generator model and its implementation. Section \ref{numerical experiments} covers our numerical study with experiments on fractional Brownian and empirical data. Section \ref{conclusion} concludes.

\section{Market Generators}\label{market generators section}

This section briefly covers the background on market generators, related papers and their commonalities that lead to the questions posed in this paper.

The potential benefits of generative machine learning for modeling financial time series have been early recognized by \citet{kondratyev2019market}, \citet{henry2019generative}, and \citet{wiese2020quant}, who first applied restricted Boltzmann machines (\citet{smolensky1986information}) and generative adversarial networks (\citet{goodfellow2020generative}) respectively to financial sequences data. Since those papers, a host of use cases have been proposed that include time series forecasting (\citet{wiese2020quant}), trading strategy backtesting (\citet{koshiyama2021generative}), hands-off volatility surface modeling (\citet{cont2022simulation}), option pricing and so-called \textit{deep hedging} (\citet{buehler2020data}, \cite{buehler2021deep}), and more.

Closest related to this work is the paper of \citet{buehler2020data} who also use the combination of variational autoencoder and signatures, with the general aim of reproducing some of the stylized facts of financial time series (\citet{cont2001empirical}). However, they use signatures in the input space and then output signature values, which implies that deploying their model requires a scalable inversion of signatures back to paths, which is far from trivial. Similar to \citet{ni2020conditional} who use signatures in the discriminator of a GAN architecture, the generator in this paper does not output signatures but only uses them in the loss evaluation step. Moreover, much work on market generators has revolved around adjusting the loss function such that desired features of the timeseries are reproduced. This is very close to our proposed approach. \citet{cont2022tail} evaluate the tail risk (value-at-risk and expected shortfall) of the simulations to evaluate their adequacy. Recently, \citet{fingan2023} included profit-and-loss (P\&L) similarity and Sharpe ratios in the loss function to increase the financial realism of the generated scenarios. This paper builds on top of that. Drawdown is a similar metric that in most financial contexts would be a useful feature to reproduce, because it captures both P\&L and autocorrelation, but in fact is in some financial contexts the most important metric, such as portfolio drawdown control, optimization, or drawdown insurance. For instance, some investment strategies or fund managers get automatically closed down if they breach certain drawdown limits and it is thus crucial that a simulation of their strategies does not understate the probability of larger drawdowns. However, as a path functional, rather than a function on a P\&L distribution like a tail loss or Sharpe ratio, drawdown is much more difficult to evaluate. That is why we introduce the approximation trick in the next sections.

\section{Drawdowns}\label{section drawdowns}

This section introduces the concept of a drawdown and its notation used throughout the paper. We briefly touch upon how it is usually approached in the literature, the issues for data-driven simulation of drawdowns and the expression of drawdown as a non-linear dynamic system. The latter insight, combined with the concept of a path signature, will allow us to introduce the approximation in the next sections.

\subsection{Introduction and notation}
The drawdown is the difference at a point in time $t$ between a price $S_t$ and its historical maximum up to that point. For a price path (in levels) $S: [0, T] \xrightarrow{} \mathbb{R}$ define drawdown function $\Xi$ as:
\begin{equation}\label{drawdown definition}
    \xi_t = \Xi(S)_t = \max(\max_{k<t}(S_k) - S_t, 0) 
\end{equation}

We are interested in the distribution $\mathcal{P}(\xi)$ as a function of the DGP parameters of $S$, denoted by $\theta$. It is clear that the drawdown is a non-trivial function of the underlying DGP and captures at least three dimensions. Firstly, the drift or deterministic tendency of achieving a new maximum. Secondly, the dispersion or volatility of $S$, which determines the probability of a loss vis-a-vis this monotonic drift. Thirdly, and not grasped by standard Markovian assumptions, the autocorrelation of losses or the probability that the nonzero drawdown will persist, i.e. the duration of the drawdown. The analytical link between these components or a generalized DGP and the drawdown distribution is unknown and depends on the specific process. In the drawdown literature that uses parametric simulation (e.g. \cite{carr2011maximum}\cite{atteson2022carr}\cite{rej2018you}\cite{douady2000probability}), one often assumes (geometric) Brownian motion and leverages Lévy’s theorem (\citet{levy1940certains}), which determines the joint law of the running maximum of a Brownian and the deviation to this maximum, to find the analytical distribution of drawdown as a function of the drift and volatility of the DGP. This yields closed-form expressions of the distribution $\mathcal{P}(\xi)$ (e.g. see \citet{douady2000probability} and more concisely \citet{rej2018you}). However, martingale processes like Brownian motion do not exhibit return autocorrelation and one expects a sequence of consecutive losses to result in larger drawdown scenarios. \citet{goldberg2017drawdown} assume an autoregressive process of order 1 (AR(1)) and show that increasing autocorrelation leads to more extreme drawdowns. \citet{van2020drawdowns} introduce \textit{drawdown greeks} and discuss the sensitivity of the maximum $\xi$ value to the autocorrelation parameter under normal return assumptions and find strong dependence on the assumed level of AR(1) autocorrelation. Hence, standard martingale assumptions  only lead to optimistic lower bounds on those worst case drawdown scenarios, as was also emphasized in \citet{rej2018you}. Obtaining a simulated $\mathcal{P}(\xi)$ faithful to the empirical drawdown distribution through selecting DGP parameters is thus hard to obtain, as it is very sensitive to serial dependence and the analytical link between $\theta$ and $\mathcal{P}(\xi)$ does not exist except under very restrictive assumptions (cf. \cite{van2020drawdowns}). In drawdown papers that use non-parametric simulation (e.g. \cite{chekhlov2004portfolio}\cite{chekhlov2005drawdown}\cite{molyboga2016portfolio}\cite{van2020drawdowns}), the scenarios are limited to the historical sample (e.g. historical blocks in a block bootstrap procedure), which drastically limits available samples. This jeopardizes the convergence of data-hungry models in the non-overlapping block case, and creates multicollinear (or identical) conditions for the overlapping (or oversampled) blocks case. For instance, \citet{chekhlov2005drawdown} discuss how multiple scenarios increase the effectiveness of drawdown optimizers over a single historical scenario, but also indicate how oversampling a limited dataset has diminishing returns in terms of out-of-sample risk-adjusted returns of the drawdown optimization strategy. Ideally, one has a vast number of rich drawdown scenarios, but this is limited with historical simulation and an unsolved problem for parametric Monte Carlo. One way to cope with these issues is to simply not embed the parameters that capture drawdown in the DGP, derive its distribution from said DGP assumptions or rely on few, overlapping or identical paths, but to construct a non-parametric simulation approach that relies on learning with flexible mappings rather than calibrating known parameters. We show that one can rather rely on a market generator to abstract both the DGP and the link between DGP and drawdown distribution, and learn to reproduce this distribution in a Monte Carlo. Learning means updating the parameters of an implicit DGP (e.g. the parameters of a neural network), denoted by $\theta$, over batches of training data of $S$ to increasingly improve its ability to reproduce the drawdowns in the synthetic samples, denoted by a parametrized path $S_\theta$. This Monte Carlo could be used to devise strategies that control drawdown or as a simulation engine for pricing drawdown insurance in a fully non-parametric way.

\subsection{Issues with the drawdown measure for data-driven simulation}
The problem with the drawdown measure, compared to simpler measures on the moments of the P\&L, is that at first glance it looks unsuited for learning. The reasons are twofold.
\begin{itemize}
    \item \textbf{Differentiability}: $\Xi$ is non-differentiable w.r.t. a parameterized path $S_\theta$. At first glance it looks impossible to change the DGP parameters $\theta$ by including 'feedback' on $\frac{\delta(\Xi(S_\theta))}{\delta(\theta)}$, without making simple assumptions on the specification of $\theta$.
    \item \textbf{Complexity}: evaluating the maximum of a vector of $n$ numbers takes $n$ operations or $O(n)$. Naively evaluating the running maximum takes $n(n+1)/2$ or $O(n^2)$ operations. If one would use a  smoothed approximation of the maximum (such as smooth \textit{normalized exponentials} \cite{bishop2006pattern} or \textit{softmax transformation}) to resolve non-differentiability, this would be computationally prohibitive (especially for long $n$ paths). Moreover, accuracy of such naive exponential smoothing would rely on the scale of $S$ and variation around the local maxima.
\end{itemize}

The main contribution of this paper is that these issues of non-differentiability and numerical complexity can be jointly overcome by approximating the drawdown of a path by using linear regression on its signature terms.

\subsection{Drawdown as a non-linear dynamic system and Lipschitz continuity}\label{section nl ds lip reg}

Assume the price path $S : [0, T] \xrightarrow{} \mathbb{R}$ is a piecewise linear continuous path of bounded variation (e.g. finite variance), such as interpolated daily stock prices, index levels or fund net asset values (NAV)\footnote{For the definition of boundedness in the \textit{p-variation} sense see \cite{lyons2014rough}.}. Denote by $\mathcal{V}([0,T], \mathbb{R})$ the (compact) space of all such paths. Firstly, we need to remark that while from (\ref{drawdown definition}) we know the max operator makes $\Xi$ not continuously differentiable, the variation in $\Xi$ is bounded by the variation in the paths. When taking two bounded paths the distance between their maximum values is bounded by a norm on the distance between their path values:
\begin{equation}
        |\max_{0<i<T}(S^1_i) - \max_{0<j<T}(S^2_j)| \leq \max_{i,j \in [0,T]}|S^1_i - S^2_j| = ||S^1 - S^2||_{\infty}
\end{equation}
where $|.|$ denotes the standard absolute value sign (as an $S_t$ is a scalar), while $||.||_{\infty}$ denotes the infinity norm or maximum distance between two paths (as an $S$ is a full path). In words, the distance between two maxima is capped by the maximum distance between any two values on path $S^1$ and $S^2$. This means that if two paths become arbitrarily close in terms of this distance, their respective maxima will become arbitrarily close. More specifically, the maximum is Lipschitz-$C$ continous, with distance inf-norm and $C=1$. Similar arguments can be made for $\Xi$, i.e. application of the maximum operator and a linear combination of two piecewise linear paths is Lipschitz continuous. In particular, Proposition 3.1 shows that the impact on the drawdown function of a change in the underlying paths is bounded in terms of a defined distance metric (inf-norm).

\begin{proposition}[\textbf{\textit{Lip}-regularity of $\Xi$}]\label{Prop 3.1}
Consider by $\mathcal{V}([0,T], \mathbb{R})$ the space of continuous paths of bounded variation $[0,T] \rightarrow \mathbb{R}$, two paths $S^1, S^2$ $\in \mathcal{V}$ and drawdown function $\Xi(S)_t : \mathcal{V} \rightarrow \mathbb{R} = \max_{k \leq t}(S_k) - S_t$. We have the following Lipschitz regularity for inf-norm distance $||.||$ and a regularity constant $C$:
\begin{equation}
    ||\Xi(S^1)_t-\Xi(S^2)_t|| \leq C ||S^1 - S^2|| 
\end{equation}
\end{proposition}
\begin{proof}
\begin{equation}
    \begin{split}
        ||\Xi(S^1)_t-\Xi(S^2)_t|| = ||\max_{i \leq t}(S^1_i) - S^1_t - \max_{j \leq t}(S^2_j) + S^2_t|| \\ 
    \leq \max_{i,j \in [0,t]}||S^1_i - S^2_j|| + ||S^2_t - S^1_t|| \leq C ||S^1 - S^2||  
    \end{split}
\end{equation}
Through triangle inequality, the Lipschitz condition holds for $C$ minimal $2$ and distance metric $||.||_{\infty}$, which concludes the proof.
\end{proof}

The difference in drawdown between two paths is thus bounded by the distance between the paths. This means that if two paths become arbitrarily close according to said distance metric, their drawdowns will become arbitrarily close. In this article, we explore what this regularity implies for local approximation. 

Observe below the differentials of $\xi$ and $S$ at a specific point in time $t$, i.e. we treat drawdown as a non-linear dynamic system which is unnatural as non-differentiability implies $\frac{d\xi_t}{dS_t}$ is not a continuous function (e.g. what one would depend on for Taylor series-like local approximations).
Consider the following dynamic system (also see \cite{douady2000probability}):
\begin{equation}\label{f definition}
  d\xi_t = f(\xi_t, dS_t) = \left\{
    \begin{aligned}
      & -dS_t, \xi_t > 0\\
      & \max(0, -dS_t), \xi_t = 0
    \end{aligned}
  \right.
\end{equation}
Depending on the current level of drawdown the effect of a price change is either linear or none, hence the derivatives are path dependent.  This conditionality illustrates the non-differentiability of $\Xi$ at a time $t$ and the fact that the differentials $|d\xi_t|\leq|dS_t|$ are bounded. This observation also provides intuition as to why local approximation is feasible. It does make sense to do an interpolation between linear and zero effects, and because of the boundedness one would not be arbitrarily off. For instance, one could assume linear dynamics but with a stochastic component derived from the average time in drawdown, i.e. the probability of a linear effect. In practice, however, the solution of this stochastic equation would still depend on our estimate of the average time in drawdown and in that sense not resolve the inherent path dependence. Hence, ideally we would express $\xi$ as some function of the vector-valued $S$ (or intervals of $S$), rather than a scalar $S_t$ (or increments $dS_t$) and treat the path dependence in a more natural way. In the next section, this is exactly what rough path theory allows us to do: solving the types of equations Eq. (\ref{f definition}) belongs to, offering strictly non-commutative (thus order or path dependent) solutions to these complex non-linear dynamic systems, allowing for a unique solution for the outputs $\xi_t$ given inputs $S$, even if their effects $f$ are not continuously differentiable. 
\clearpage

\section{Rough path theory and the approximation of functions on paths}\label{rough path theory section}

This section briefly recapitulates some of the central ideas of rough path theory, path signatures and the approximation of functions on paths. Next, we discuss the signature approximation of drawdown where we consider the drawdown of a path as an approximate linear weighing of the moments of the path. This will offer the foundation for generating weighted signatures in the market generator of Section \ref{market generator}.

\subsection{Path signatures}

Rough path theory was developed by Terry J. Lyons (\citet{lyons2007differential}, \cite{lyons2014rough}, \cite{lyons2022signature}) and concerns solving rough differential equations. 
Consider the controlled differential equation (\citet{lyons2007differential}, Eq. (2)): 
\begin{equation}\label{cde example}
    dY_t = g(Y_t)dX_t
\end{equation}
where $X$ is a path of bounded variation, called the driving signal of the dynamic system. $g$ is a mapping called the physics that models the effect of $dX_t$ on the response $dY_t$. A controlled differential equation distinguishes itself from an ordinary differential equation in that the system is controlled or driven by a path ($dX$) rather than time ($dt$) or a random variable (stochastic differential equation, $d\varepsilon$). Rough path theory considers solution maps for driving signals that are much rougher (highly oscillatory and potentially non-differentiable) than e.g. a linear path of time or a traditional Brownian driving path. It is more robust to consider Eq. (\ref{cde example}) over Eq. (\ref{f definition}), and replace $dS_t$ as an input with its integral. We can rewrite  Eq. (\ref{f definition}) in this form by setting $Y = \Xi$, $X_t = (t, \int_{s=0}^{t}dS_sds)$ and $g(y, (t,x))=f$ (e.g. see \citet{liao2019learning}, Eq. (2)). This allows the effect to be of a broader type and need not even be differentiable for equation (\ref{cde example}) to be well defined.
The Picard-Lindelöf theorem (\citet{lyons2007differential}, Theorem 1.3) states that if $X$ has bounded variation and $g$ is Lipschitz\footnote{For a complete definition of this regularity assumption, see \cite{lyons2007differential}, Definition 1.21. This means (1) the effect $g$ being bounded over all values $y$ in the image of $Y$. This naturally holds for drawdown of bounded variance paths. And (2) the difference in effect of two inputs is bounded by the difference in inputs, which was the main result in Section \ref{section nl ds lip reg}}, then for every initial value $y_0$ in the image of $Y$, the differential equation (\ref{cde example}) admits a unique solution. Importantly, if the effect of $X$ on $Y$ is not Lipschitz, we lose the uniqueness of the solution. As shown in Appendix \ref{appendix}, through Picard iteration and under an additional regularity assumption\footnote{For simplicitly, in analogy to \citet{lyons2014rough}, we assumed here that the iterated $g^{\circ m}$ takes their values in the space of symmetric multi-linear forms, but this generalizes to any Lipschitz continous $g$ as per Remark 1.22 in \citet{lyons2007differential}. This simply means that if $g$ is differentiable, the $g^{\circ m}$ are indeed the classical $m$-th order derivatives of $g$, in general they are only polynomial approximations of $g$ at increasing orders.} on $g$, one naturally arrives at $M$-step Taylor approximation $\hat{Y}(M)_t$ on the path space for the $Y_t$ in Eq (\ref{cde example}):
\begin{equation}\label{Taylor expansion}
        \hat{Y}(M)_t = y_0 + \sum_{m=1}^{M}g^{\circ m}(y_0)\underset{u_1<...<u_m, u_1,...,u_m \in [0,T]}{\int...\int}dX_{u_1}\otimes...\otimes dX_{u_m}
\end{equation}

The $m$-fold iterated integral on the right hand side of Eq (\ref{Taylor expansion}) is called the $m$-th order \textit{signature} of the path $X$, which we denote by $\Phi^m(X)$. 

Signatures have useful symmetries and properties, for our use case most notably \textit{universality} and \textit{uniqueness}, which we will briefly explain below. For now, we should note that equations of this form admit uniform estimates of $Y$ when $X$ is a rough path (highly oscillatory and non-differentiable), and it suffices to control the variation of the path $X$ (of bounded $p$-variation, such as finite-variance paths) and the regularity of $Y$ (\citet{lyons2007differential}, conditions\footnote{In short, the Lipschitz continuity of $Y$ as per above implies the continuity of the Ito-Lyons solution map.} for Eq. (2) to hold).

Clearly, $\hat{Y}(M)_t$ is linear in the truncated signature of $X$ up to order $M$. Moreover, the error bounds of $\hat{Y}(M)_t$ to approximate $Y_t$ display a factorial decay in terms of $M$:
\begin{equation}\label{factorial decay formula}
    |Y_t - \hat{Y}(M)_t| \leq C_{\gamma}\frac{|X|^{M+1}_{1, [0,t]}}{M!}
\end{equation} 
where the rate decay depends on the roughness of the signal\footnote{The roughness in the \textit{p-variation} sense can be understood as the maximum (supremum) amount of variation (bounded by a defined $L^p$-norm, $|X|^p_{u, [0,T]}$) a path $X$ over $[0, T]$ can have over all subdivisions $u \subseteq [0,T]$. The more complex the path, e.g. antipersistent variation, the worse the approximation becomes for a given level of $M$.}, denoted by $\gamma$ \cite{boedihardjo2015uniform}. Eq. (\ref{Taylor expansion}) and (\ref{factorial decay formula}) aptly summarize the signature's \textit{universality}, i.e. \textit{a non-linear $Y$ on $X$ can be well approximated by a linear function $L$ on $\Phi^m(X)$}.

For instance, if one replaces the path $X$ by scalar $x$ and thus non-commutative tensor products $\otimes$ become simple products, the $m$-times repeated integral on $dx$ is $\frac{x^m}{m!}$ which is the classical exponential series we know from a standard Taylor expansion. Signatures are thus the monomials of the path space (\citet{lyons2014rough}), the generalization of local approximation to functions of paths by offering a strictly non-commutative exponential function. 

\citet{chevyrev2018signature} introduce signatures as the moment generating function on the path space, as they play a similar role to normalized moments for path-valued data instead of scalar-valued data. This is the \textit{uniqueness} property: \textit{paths with identical signatures are identical in the distributional sense as are scalars with the same sequence of normalized moments. }

%\footnote{The proof is extremely concise, and that is why we summarize it here. Since $\Phi(X)$ is a tensor algebra of X, the family of all linear combinations of $\Phi(X)$ is an algebra. This can be seen from the shuffle product property, as explained in \cite{chevyrev2016primer}. Constant functions are preserved by the constant zeroth term of $\Phi(X)$. The algebra separates the points as comes naturally from \cite{lyons2007differential} Corollary 2.16. Stone-Weierstrass then yields that the linear family on $\Phi(X)$ is dense in the space of continuous functions on X.}

Proof of signature universal non-linearity is due to \citet{lemercier2021distribution} (also see \citet{lyons2022signature} for a more recent redefinition of the result in Theorem 3.4) who prove that for a Lipschitz\footnote{Again this means $h$ does not need to be $k$-times differentiable, but the $k$-th order differentials need to be bounded.} continous $h$, path of bounded variation $X$, there exist a vector of linear weights $L$ such that for any small $\epsilon$:
\begin{equation}\label{sig approx}
    |h(X) - \langle L,\Phi(X) \rangle| \leq \epsilon
\end{equation}
where $\langle x_1,x_2 \rangle$ denotes the standard inner product between vectors $x_1$ and $x_2$, and $\Phi(X)$ the infinite $M$ signature or infinite collection of iterated integrals of a path. By the Stone-Weierstrass theorem (a crucial theorem in proving the universal approximation capabilities of neural networks, \citet{cotter1990stone} and \citet{hornik1991approximation}) it is proven that signatures are a universal basis in the sense that they allow us to express non-linear path functions $h$ as a linear function of signatures $\Phi(X)$ provided that they have the required regularity. As for classical Taylor series, although a function might not be $k$-times differentiable, a high order smooth polynomial approximation can be a quantifiably close (bounded error) approximation for bounded Lipschitz functions, provided $C$ is small. In this paper, we apply this common insight to path functions, and drawdowns in particular.

\subsection{Signature approximation of drawdown}

We noted in Section \ref{section drawdowns} that drawdown is a non-linear, non-differentiable function of its underlying path $S_{\theta}$, which is smooth in the bounded differentials $d\xi/dS$ sense. Its outcome can be seen as an interpolation between two types of path dependent effects. In this section, we leverage the boundedness of Section \ref{section drawdowns} together with the universality of signatures of Section \ref{rough path theory section} to introduce a smooth local approximation by linear approximation of drawdown on path signatures.

We propose an approximation $\hat{\xi}(M)_t$ of $\xi_t$ of the form:
\begin{equation}\label{drawdown approx full}
    \hat{\xi}(M)_t = \xi_0 + \sum_{m=1}^{M}L_m\underbrace{\underset{u_1<...<u_m, u_1,...,u_m \in [0,t]}{\int...\int}dS_{u_1}\otimes...\otimes dS_{u_m}}_{\Phi^m(S)}
\end{equation}
where $L_m$ is a vector of linear coefficients linking the drawdown at $t$ with the signature terms of order $m$ of the path $S$ up to $t$. As per above, $L$ could be considered the iterated effects of intervals up to $t$ on the resulting drawdown $\xi_t$, where the coefficients are \textit{not} the $m$-th order derivatives as they are not defined, but polynomial approximations that are essentially numerical interpolations of the nested effects. Signatures thus offer a strictly non-commutative alternative to the polynomials Taylor series would suggest. 

Importantly, for this integration to make sense the path $S$ has to be continuous, hence in practice augmented from discrete observations into the continuous domain by \textit{time-augmenting} the path. This means adding time as an axis and assuming piecewise linear paths\footnote{Note that in the computation of signatures of piecewise linear paths, one can use Chen's identity \cite{lyons2014rough} to compute the signature as the iterated tensor product of the increments of the path along the time axis, which alows for efficient computations for practical discrete (but assumed piecewise linear) data (cf. practical computation of signatures, Section 5 in \cite{lyons2022signature}).}. Other options such as \textit{lead-lag} and \textit{rectilinear} augmentation exist (see \citet{lyons2022signature}), but are not favored for this application and may add dimensionality to the paths which increases the number of signature terms per level $M$ and related compute time. 

As the ordered iterated integrals represent the drift, Levy area, and higher order moments of the path distribution (see \citet{chevyrev2018signature}), Eq. (\ref{drawdown approx full}
) thus argues that \textit{drawdown can be approximated as a linear function of the moments of the path}\footnote{Already note here the parallel with the link between quantiles and traditional moments, which we will restate in Section \ref{model intro section}.}. Leveraging factorial decay of the approximation error for Lipschitz functions, we argue with Proposition (\ref{Prop 3.1}) and Eq. (\ref{f definition}) that with the full signature $M \xrightarrow{} \infty$ one gets an arbitrarily close approximation of  $\xi$, where the rate decay depends on the roughness, denoted by $\gamma$, of the underlying price process $S$ (see proofs in \citet{boedihardjo2015uniform}).

In the more compact inner product notation, we propose to apply Eq. (\ref{sig approx}) to drawdown:
\begin{equation}\label{x}
    |\Xi(S) - \langle L,\Phi(S) \rangle| \leq \mu
\end{equation}
where the arbitrary precision $\mu$ is only in theory, because in practice we rely on the truncated signatures of level $M$ as the full signature is an infinite collection, and thus there will be an error $\kappa$ that due to the ordered nature of the coefficients decays factorially in $M$ (e.g. Eq. \ref{factorial decay formula} for the intuition and \citet{boedihardjo2015uniform} for the proofs):
\begin{equation}\label{ox}
\begin{split}
    \Xi(S) = \langle \hat{L},\Phi^M(S)\rangle + \kappa_M
\end{split}
\end{equation}
\begin{equation}
    \hat{\Xi}_M(S) = \langle \hat{L},\Phi^M(S)\rangle
\end{equation}
where $\hat{L}$ are the estimated coefficients for a chosen signature truncation level $M$, contrasting the theoretical infinite collection of weights $L$. Similarly, $\hat{\Xi}_M(S)$ is the approximated drawdown for this truncation level $M$, while $\Xi(S)$ is the exact value. Note that one could also do an equivalent truncation of the number of linear coefficients $len(\hat{L})$ rather than the signature order. However, from Eq. (\ref{drawdown approx full}) we know that it is more natural to choose a set of linear coefficients that corresponds to a number of signature terms following the choice of M.

Proposition 4.1 looks into the consistency behaviour of $\kappa$ with respect to a sample size of $K$ sample paths drawn from $\mathcal{V}$ and next we highlight small sample properties that become apparent from the proof.

\begin{proposition}[\textbf{Consistency of linear $\mathbf{\Xi(S)}$ approximation on signatures $\mathbf{\Phi^M(S)}$}]\label{Prop 4.1} 
Consider by $\mathcal{V}([0,T], \mathbb{R})$ the space of continuous paths of bounded variation $[0,T] \rightarrow \mathbb{R}$, $\mathcal{K} \subset \mathcal{V}$ is a compact subset comprising sample paths $S_k$, $k \in [0, ..., K]$, and $\Xi : \mathcal{V} \rightarrow \mathbb{R}$ is the Lip continuous drawdown function.
The approximation error $\kappa$ of $\Xi(S)$ for \textit{any} $S$ in $\mathcal{V}$ by $\hat{\Xi}(S)$ is bounded through the regularity of $\Xi$ (Proposition 3.1) and the distance in the signature space between $S$ and any $S_k$, such that for $K \rightarrow \infty$, $\kappa \rightarrow 0$, or
\begin{equation}
    |\Xi(S) - \hat{\Xi}(S)| \rightarrow 0, \textnormal{ for }  K \rightarrow \infty
\end{equation}
\end{proposition}
\begin{proof}
    Universal nonlinearity of signatures (Eq. \ref{sig approx}) is due to Theorem 2.1 in \citet{lemercier2021distribution} and states that for defined $\mathcal{V}$ and $\mathcal{K}$ there exists a truncation level $M \in \mathbb{N}$ and coefficients $\hat{L}$ such that for every $S_k \in \mathcal{K}$ we have that for any $\iota$
    \begin{equation}\label{lemercier}
        |\Xi(S_k) - \langle \hat{L}, \Phi^M(S_k) \rangle | \leq \iota
    \end{equation}
    We decompose $\kappa = |\Xi(S) - \hat{\Xi}(S)|$ with a triangle inequality as suggested by Eq. (3.6) in \citet{lyons2022signature}:
    \begin{align}
        & |\Xi(S) - \hat{\Xi}(S)| \leq |\Xi(S) - \Xi(S_k)| + \\
        & |\Xi(S_k) - \hat{\Xi}(S_k)| + |\hat{\Xi}(S_k) - \hat{\Xi}(S)|  \\
        & = \textbf{A}_K + \textbf{B}_{M(K)} + \textbf{C}_K 
    \end{align}
    This inequality bounds the error by the regularity of $\Xi$ in $\textbf{A}$, the $\iota$ of Eq. (\ref{lemercier}) in $\textbf{B}$ and a signature distance in $\textbf{C}$.

By Proposition 3.1, we find that $\textbf{A} \leq C||S-S_k||_{\infty}$. For any $k \in [1, ..., K]$ with $K \rightarrow \infty$ and the compactness of $\mathcal{V}$ there exists a $k$ such that $||S - S_k||_{\infty} \rightarrow 0$ such that $\textbf{A}$ can be reduced to zero.
By Eq. (\ref{lemercier}), $\textbf{B} \leq \iota$, such that this term is governed by the rate decay of Eq. (\ref{factorial decay formula}). Eq. (\ref{lemercier}) guarantees we can pick an $M$ high enough such that through Eq. (\ref{factorial decay formula}), $\iota \leq C_{\gamma}\frac{|S|^{M+1}_{1, [0,T]}}{M!}$ such that this term can be shrunk arbitrarily small as $M$ can be set arbitrarily large for $K \rightarrow \infty$.

As per Eq. (3.7) in \citet{lyons2022signature} the difference between the approximated drawdown of two paths can be bounded by a linear combination of the difference in signatures:
\begin{equation}
    \textbf{C} \leq |\langle L, |\Phi^M(S_k) - \Phi^M(S)| \rangle|
\end{equation}
for an appropriate choice of rough path distance metric \cite{lyons2022signature}\footnote{More specifically, an $L^p$-norm in the path space that bounds the $p$-variation of these differences (\citet{lyons2007differential}).}. For any $k \in [1, ..., K]$ with $K \rightarrow \infty$ and the compactness of $\mathcal{V}$ there again exists a $k$ such that $|\Phi^M(S_k) - \Phi^M(S)| \rightarrow 0$ such that $\textbf{C}$ could be reduced to zero, which concludes the proof.
\end{proof}

In sum, $\kappa$ converges to zero for $K \rightarrow \infty$ provided that Proposition (3.1) holds for $\textbf{A}$, while compactness of $\mathcal{V}$ ($\textbf{B}$ and $\textbf{C}$) is a reasonable assumption in practice.

For finite $K$, we can still use the decomposition in the proof to do error analysis. The drawdown approximation will generalize well as far as (A) the maximum distance between available samples $S_k$ and possible new samples in $\mathcal{V}$ is small, (B) the signature truncation level $M$ is chosen appropriately high for the given roughness of the process to let $\iota$ be small enough, (C) as per (A) the distance between the signatures of observed and unobserved paths is small such that term $\mathbf{C}$ is small.

Finally, a remaining modelling choice is the specific choice of $L$. As the approximation is essentially linear, Eq. (\ref{ox}) can be estimated using linear regression (OLS) and higher order polynomials are not required (as one would need to do with e.g. \textit{logsignature} regression \cite{lyons2022signature}). However, since in practice the sample $K$ is limited and the number of signature terms scales exponentially with $M$, one has to be mindful of overfitting on a limited sample size $K$, e.g. $len(\Phi^M) > K$. Therefore, we study the impact of regularized linear regression of  $\Xi(S)$ on $\Phi^M(S)$ with a penalty for the number (absolute shrinkage or selection) and size (proportional shrinkage) of estimated coefficients (i.e. the elastic net (ElNet) regression). We conclude by specifying:
\begin{equation}\label{approx final specification}
    \hat{L} = \min_{L}(||\Xi(S)-\langle L, \Phi^M(S) \rangle||_2 + \lambda_1||L||_1 + \lambda_2||L||_2)
\end{equation}
\begin{equation}
    \hat{\Xi}_M(S) = \langle\hat{L},\Phi^M(S)\rangle
\end{equation}
where $\lambda_1=\lambda_2=0$ corresponds to OLS, $\lambda_1=0$ corresponds to Ridge and $\lambda_2=0$ to LASSO regression \cite{hastie2009elements}. In applications below, we set $\lambda_1$ and $\lambda_2$ using 10-fold cross-validation (CV), such that we can further refer to Eq. (\ref{approx final specification}) as \textit{ElNetCV}.

\section{The model}\label{market generator}

\subsection{Introduction}\label{model intro section}

In this section, we introduce and motivate our generative ML model. The aim of the drawdown market generator, from here on named the $\xi$-VAE, is that upon convergence of (train and validation) reconstruction loss terms we have guarantees that the synthetic samples have preserved the drawdown distribution of the original samples. This is generally not the case in market generator models or financial DGPs with standard martingale assumptions (as per Section \ref{section drawdowns}). 

Crafting a DGP to obtain a certain level of drift, volatility or higher order moments is mathematically more straightforward, as those are typically the equations that constitute the DGP. With measures on the P\&L, such as value-at-risk (VaR) or expected shortfall (ES) (\citet{cont2022tail}), one can leverage the direct analytical link between quantiles and moments (e.g. \citet{boudt2008estimation}). As per Section \ref{section drawdowns}, one can express the drawdown distribution as a function of the moments of the static P\&L described by the DGP as well, but only under very restrictive assumptions (e.g. \citet{douady2000probability} and \citet{rej2018you}). Besides, one would expect drawdown to be a function of the moments of the path (vector-valued) $S$, rather than the (scalar-valued) $S_t$. 

 In this article, we proposed to \textit{approximate drawdown as a linear combination of the moments of the path}, as \citet{chevyrev2018signature}) define signatures as the moments of the path. As an analogue to quantiles and static moments, we can evaluate these path moments and weigh them according to their importance to simulate realistic drawdowns. Moreover, this approximation implies a smoothing of the drawdown function by the change of basis. The signature is a non-parametric sum of path values. Weighted sums are differentiable\footnote{That is why \citet{kidger2020signatory} focus on this differentiability property for efficient CPU and GPU implementations in their \textit{signatory} (\url{https://pypi.org/project/signatory/}) package, which we use in our Python code.}. Indeed, because of linearity the loadings of drawdown to signature terms can also be seen as the sensitivities of a path's drawdown to changing signatures. 
 
Measuring the divergence between the moments of a path by means of a maximum mean discrepancy (MMD) was proposed by \citet{buehler2020data}. We essentially add that it is useful to weigh these moments according to the $L$ from Eq. (\ref{approx final specification}) to minimize (and control for) the drawdown divergence between the input and the output samples during training (and validation) epochs. This has the advantages of (1) not requiring signatures in the input or output space, only as part of the objective, and (2) having an explicit drawdown term in the reconstruction loss that allows one to monitor its convergence during training and validation. The next subsection will make this motivation specific and introduce the algorithm.

\subsection{The algorithm}

\paragraph{\textbf{Variational autoencoders}} The core of our algorithm is a variational autoencoder (VAE), which is a general generative machine learning architecture that links an encoder and decoder neural network in order to generate new samples from a noisy, in this case Gaussian, latent space. The idea of a Monte Carlo is to transform noise into samples that are indistinguishable from the actual samples by scaling them, adding drift, etc. In other words, the neural network that constitutes the decoder is our non-parametric DGP. It does contain the parameters of the neural network, but it is non-parametric in the sense that we do not have to specify the dynamics in a handcrafted formula before we can do Monte Carlo. We rather rely on the universal approximation theorems behind even shallow neural networks (\citet{hornik1991approximation} and \citet{cotter1990stone}) to approximate a realistic DGP by iterating data through the network and updating the parameters $\theta$ with feedback on the drawdown distribution of the batch, assuring that the approximated DGP converges in train and validation loss to the empirical DGP.

We will not discuss the VAE architecture in depth here, but include an architectural overview in Appendix \ref{Appendix 2: variational autoencoder} and refer the interested reader to \citet{kingma2014stochastic} for details on the encoder and decoder networks $g$, backpropagation, the latent Kullback-Leibler $\mathcal{L}_L$ loss and the standard $L2$ reconstruction loss $\mathcal{L}_R$. 

We should stress here that the main reason for picking a rather standard VAE (over restricted Boltzmann machines, generative adversarial networks, generative moment matching networks or normalizing flow-based VAEs) is their  simplicity, speed, flexibility, scalability and stability during training. Boltzmann machines are very efficient to train, but the energy-based loss function and their binary values makes them very inflexible for adjusting objective functions. Through the discriminator mechanism GANs are most flexible and a very popular choice in related literature, but notoriously expensive to train in terms of required data and speed, and associated instabilities such as \textit{mode collapse} and \textit{vanishing gradients} leading to subpar results (\citet{eckerli2021generative}).

\begin{algorithm}[h!]
\scriptsize
\caption{Training $\xi$-VAE}\label{euclid1}
\label{"Algorithm 1"}
\hspace*{\algorithmicindent} \textbf{Input} \multiline{Historical price paths $S: [0,T] \rightarrow \mathbb{R}$, hyperparameters (listed in Appendix \ref{Appendix 2: variational autoencoder}), signature truncation level M and feature weight $\alpha$.} \\
\hspace*{\algorithmicindent} \textbf{Output}{ Trained VAE Market Generator $g_{\theta}$} \\ 
\begin{algorithmic}[1]
\Procedure{Train}{}
\State \multiline{
Divide historical sample into blocks (index $b$, $b \in [1, T-\tau]$) of length $\tau$, calculate the signatures of these paths truncated at level $M$, $\Phi^M(S_b)$, calculate the drawdowns $\Xi$ of these paths $\Xi(S_b)_{\tau} = \int_{0}^{\tau}(max_{t_i < t}(S_{b,{t_i}}) - S_{b}, t)dt$
}

\State\label{Step 3} {$\hat{L} \xleftarrow{ }ElNetCV({\Xi}(S_{b}), \Phi^M(S_b)$}
\State {Initialize the parameters $\theta$ of the VAE.}
\For {$i: \{1,...,N_t\}$}:
\State \multiline{
Sample a batch (index $\mathcal{B}$) of blocks and pass it through the encoder $g_{\theta}$ and decoder network $g^{-1}_{\theta}$.}
\State\multiline{Calculate drawdown $\Xi(S')$ of the output sample $S'$ using the differentiable signature approximation: $\langle \hat{L}, \Phi^M(S') \rangle$ }
\State\label{Step 8}\multiline{Define the reconstruction loss term as the weighted average of $L2$ error and drawdown loss:
$\mathcal{L_R} = \mathbb{E}_{\mathcal{B}}||S-S'||^2 + \alpha\mathbb{E}_{\mathcal{B}}||\langle \hat{L}, \Phi^M(S) \rangle - \langle \hat{L}, \Phi^M(S') \rangle||^2$}.
\State{$\mathcal{L} = \mathcal{L}_L + \mathcal{L}_R $}
\State\label{Step 10}{$\theta_i \xleftarrow{} \theta_{i-1} - l \frac{\delta\mathcal{L}(\theta)}{\delta\theta}$}
\EndFor
\EndProcedure
\end{algorithmic}
\end{algorithm}
\normalsize

\begin{algorithm}[h!]
\scriptsize
\caption{Sampling from $\xi$-VAE}\label{euclid2}
\label{"Algorithm 2"}
\hspace*{\algorithmicindent} \textbf{Input} Trained VAE Market Generator $g_{\theta}$. \\
\hspace*{\algorithmicindent} \textbf{Output} $N_g$ generated samples $S'$  \\ 
\begin{algorithmic}[1]
\Procedure{Generate}{}
\For{$j: \{1, ..., N_g\}$}
\State{Sample a random Gaussian variable Z}
\State{$S' \xleftarrow{} g^{-1}_{\theta}(Z)$}
\EndFor
\EndProcedure
\end{algorithmic}
\end{algorithm}
\normalsize

\paragraph{\textbf{Training and sampling from $\mathbf{\xi}$-VAE}} The proposed algorithm is provided in Algorithm \ref{"Algorithm 1"} and \ref{"Algorithm 2"}. In short, we propose to include the divergence between the observed drawdown distribution of a batch $\mathcal{B}$, which is a stochastically sampled set of blocks of $\tau$ subsequent points of $S$, and the synthetic drawdown distribution (the drawdowns of reconstructed samples $S'$) in the reconstruction loss function. The market generator can be interpreted as a moment matching network, adding the moments of drawdown rather than returns\footnote{For instance, this distributional distance over a batch is identical to the distance between the actual and fitted distributions in Figure \ref{Drawdown distribution}, which is the drawdown distribution we want to preserve in synthetic samples.}. 

As input $\xi$-VAE takes historical price paths $S$, a signature truncation level $M$, objective scale $\alpha$\footnote{Through Grid Search, $1e^{-4}$ was chosen for $\xi$-VAE, while zero corresponds to a standard $VAE$.}, and the VAE hyperparameters listed in Appendix \ref{Appendix 2: variational autoencoder}. The output is a trained neural DGP (encoder and decoder network) that allows to transform random Gaussian variables into new paths indistinguishable\footnote{With regard to the train and test convergence criteria $\mathcal{L}_\mathcal{R}$ ($\xi$-VAE and VAE) and $\mathcal{L}_\mathcal{\xi}$ (only $\xi$-VAE.).} from original data.

The key steps for training a $\xi$-VAE, distinctive from a standard VAE (cf. Appendix \ref{Appendix 2: variational autoencoder}), are:
\begin{itemize}
    \item \textbf{Signature drawdown approximation} (Step (\ref{Step 3})): compute the signatures up to order $M$ of each block, $\Phi^M(S_b)$, and regress them on the drawdown of the path, $\Xi(S_b)$ , using Eq. (\ref{approx final specification}). This results in one set of weights $\hat{L}$ that are used over the $N_t$ training steps. This approximation overcomes the numerical complexity issue inherent to naive iterative (closed-form $\Xi$) evaluation of drawdown.
    \item \textbf{Drawdown divergence evaluation} (Step (\ref{Step 8})): drawdown divergence over a particular batch of data is defined as the distance between the original and replicated drawdown distribution: $\mathcal{L_\xi}=\mathbb{E}_\mathcal{B}||\langle \hat{L}, \Phi^M(S) \rangle - \langle \hat{L}, \Phi^M(S') \rangle||^2$. The approximation overcomes the issue that closed-form $\Xi$ evaluations cannot yield informative gradients in Step (\ref{Step 10}) and $\mathcal{L}_{\xi}$ would be ignored (flat) in training.
\end{itemize}

While standard ($\alpha=0.0$) $\mathcal{L}_{R}$ terms converge sooner, the additional $\mathcal{L}_\xi$ term shifts the distributions towards a $\mathcal{P}(\xi)$ close to the empirical one. As per below, this can be seen by monitoring the $\mathcal{L_\xi}$ term during training and validation epochs.

\section{Numerical results}\label{numerical experiments}

This section comprises the numerical results of the outlined methods. First, we discuss the accuracy of linear approximation of drawdown in the signature space on simulated and real world data, focusing on the error rate decay and consistency. Second, we discuss the accuracy of the drawdown market generator.

\subsection{Linear drawdown approximation with signatures}

This part analyses the accuracy of the approximation. Section \ref{bottom up} describes the simulation set up and results for fractional Brownian simulated data. Section \ref{real world approx} discusses the accuracy of the approximation on empirical data.

\subsubsection{Bottom-up simulations}\label{bottom up}

\paragraph{\textbf{Simulation set up}}
We expect to find uniform decay as a function of truncation level $M$ of estimated drawdown approximation error $\hat{\kappa}$ in-sample, where the decay constant is a function of the roughness of the underlying process. Out-of-sample we rely on the error decomposition in Proposition \ref{Prop 4.1} and expect the error to shrink if $K$ grows very large. Since the roughness of empirical data is unknown and has to be estimated, it is useful to test the  approximation in an experimental set up where $\gamma$ can be specified. We thus first test the approximation (\ref{approx final specification}) on simulated fractional Brownian motion (fBM) paths.

Consider first the simplest case of homoskedastic BM ($dS_t = \mu dt + \sigma d\varepsilon$). In this simple case, the price path $S$ can thus be seen as the cumulative sum process of a random uncorrelated standard Gaussian $\varepsilon$, scaled with $\sigma$ and added a deterministic drift $\mu$. We consider piecewise linear paths of length $T=20$ days with values $\mu=1\%/252$ and $\sigma=20\%/252$. fBM implements BM where the uncorrelated Gaussian increments are replaced by fractional Gaussian increments that have a long-memory structure. The martingale property that the autocovariance between Gaussian increments has expectation zero, is replaced by a generalized autocovariance function for two increments $dS^H$ at t and s (i.e. lag $t-s$):
\begin{equation}
    E[dS^H_t dS^H_s]=\tfrac{1}{2}[|t|^{2H}+|s|^{2H}-|t-s|^{2H}]
\end{equation}
where H is the so-called Hurst exponent ($H$). Note that a $H=0.5$ corresponds to Brownian increments, while $H>0.5$ yields smooth, persistent, positively autocorrelated paths and $H<0.5$ yields rough, antipersistent negatively autocorrelated paths. Intuition tells us that the smaller $H$, the more granularity the path has, and the worse the approximation will become for a certain level of $M$.

There are hence three dimensions to this simulation study. We want to evaluate the estimated error $\hat{\kappa}$ as a function of (1) roughness $H$, (2) signature approximation order $M$ and (3) simulation size $K$. Therefore, we:

\begin{itemize}
    \item Vary $H$ between 0.4 and 0.7 with step size 0.05 
    \item Vary $M$ between 1 and 10 (naturally with unit steps)
    \item Repeat the experiment for $K$ in [1000, 5000, 10000, 20000, 50000]
    \item Fit regression Eq. (\ref{approx final specification})\footnote{As per above we assume piecewise linear paths by time-augmenting \cite{lyons2022signature} the paths and do standard scaling of the feature set such that the coefficients and regularization penalties make sense. The chosen $\lambda_1$ and $\lambda_2$ by $CV$ is sample specific, but was stable around $4e^{-5}$ regularization (i.e. the scale of both lambdas w.r.t. total objective value) and a $\lambda_1/\lambda_2$ ratio of $0.5$.} on $K$ samples (train) and simulate $p_{test}K$ new samples to evaluate test accuracy ($p_{test}=0.1$).
\end{itemize}

\paragraph{\textbf{Results and discussion}}

 \begin{figure*}
     \centering
     \begin{subfigure}{\textwidth}
         \centering
         \includegraphics[width=0.75\textwidth]{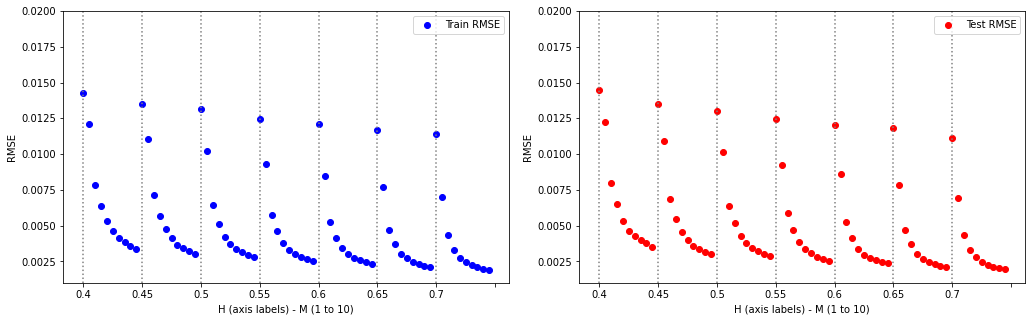}
     \end{subfigure}
     \hfill
    \quad
     \hfill
        \caption{Factorial error (RMSE) decay as a function of H}
        \label{Factorial error (RMSE) decay as a function of H}
\end{figure*}

For readibility, the detailed numerical results figures are included in Appendix \ref{Appendix 3: drawdown market generator results}, where Figures \ref{is oos discrepancy} to \ref{Train test R2 roughness} investigate the root-mean-squared ($RMSE$) approximation error of the approximation for these individual dimensions. Figure \ref{Factorial error (RMSE) decay as a function of H} is included in the main text, as it gives a good overview of the expected rate decay as a function of $H$ and $M$ and the consistency for large sample sizes $K$. 

We are initially interested in the consistency of the regressions and evaluate the in- and out-of-sample discrepancy as a function of sample size. Figure \ref{is oos discrepancy} shows the difference between in- and out-of-sample accuracy as a function of sample size $log(K)$. We averaged the difference for all simulation sizes over the values of $M$ to get a line per value of $H$ (blue line), the opposite (red line) or the overall average per $K$ (black line). 

The conclusions are twofold: (1) when $K$ grows large the discrepancy disappears, which relates to Proposition \ref{Prop 4.1} in the sense that for $K \rightarrow \infty$ the two error components (terms $\textbf{A}$ and $\textbf{C}$) that explain differences in in- and out-of-sample performance shrink to zero, and (2) the in-sample(IS)/out-of-sample(OOS) accuracy divergence depends on confounders $M$ and $H$ implying an improperly chosen $M$ for smaller and/or rougher samples might lead to bad generalization (term $\textbf{B}$ in error decomposition).
 
 Next, we look in Figure \ref{Train test RMSE order} at the train and test $RMSE$ as a function of the signature approximation order $M$. When taking the average performance over $K$ and $H$, we find uniform decay in $RMSE$. When deconfounding for sample size $K$, which you can find on the right hand side of Figure \ref{Train test RMSE order}, we find that only for the smallest sample size the improvement in test accuracy stalls before $M$ reaches its maximum value. This implies that for large samples sizes one would prefer to choose the highest signature order that is possible given computational constraints, i.e. the computing time scales quadratically with the signature order $M$ while the improvement in accuracy has diminishing returns (also see Figure \ref{Average compute time}). In sum, this indicates worse generalization abilities for smaller samples driven by the distances discussed in Proposition \ref{Prop 4.1}, rather than the $M/K$ ratio as it is small for all experiments. 
 
Next, we check the relationship between the accuracy of the approximation and the roughness $H$ of the assumed process. In Figure \ref{Train test R2 roughness}, which shows the average $RMSE$ per level of $H$, we find that the accuracies improve uniformly with $H$. 

Finally, the factorial error decay in $M$ and the dependency of its rate on $H$ is best illustrated in Figure \ref{Factorial error (RMSE) decay as a function of H}. The vertical lines separate the different levels of $H$, denoted on the x-axis, while in between $M$ increases from 1 to 10, while $K$ is fixed at the maximum $50000$. It is clear that for rougher processes the error is higher and decreases slower than for smoother processes, which we derived from Eq. (\ref{factorial decay formula}). For all levels of $H$, high orders of $M$ yield negligible approximation errors for this high $K$. Moreover, the in-sample decay generalizes well to out-of-sample error behaviour. 

 In summary, the simulation study confirms our initial expectations of rate decay and IS-OOS consistency, but raises some practical warnings as well:
 \begin{itemize}
     \item For a large number $K$ of sample paths, say $log(K)>3$, one finds uniform decay in $\kappa$ with both $H$ and $M$, in both test and training fits. In terms of Proposition \ref{Prop 4.1}, the distances between any $S$ and the $S_k$ do become smaller with $K$ and the estimated $\hat{L}$ generalize better. 
     \item For small samples ($log(K)<3$), the unbiased approximation may generalize badly, which will result in worse accuracies for $ElNetCV$. This relates to Proposition \ref{Prop 4.1} in the sense that the distance between any $S$ and $S_k$ (in inf-norm and signature terms) can be large. In any case, $M$ needs to be chosen high enough for rough paths, but with small $K$ the system might become ill-conditioned or even degenerate. High regularization will then result from CV, which comes at the cost of higher bias and lower (even IS) accuracy\footnote{The accuracy that is good enough for the application at hand of course depends on the application (e.g. for a market generator with 10\% drawdowns, improvements of some bps by the higher order $M$ might not be worth the extra computational time). Moreover, below we will deal with data sets later that are considered large (in the order of 8000 samples) from these standards.}. This issue is very dependent on the sample and its roughness, but we generally discourage the approximation for sample sizes $log(K)<3$.
 \end{itemize}

\subsubsection{Empirical data}\label{real world approx}

\paragraph{\textbf{Data description and set up}} Consider a universe of $U{=}4$ investible instruments: equity (S\&P500), fixed income (US Treasuries), commodities (GSCI index) and real estate (FTSE NAREIT index). 
We collect price data (adjusted close prices) between Jan 1989 and May 2022, which gives us T=8449 daily observations. Clearly, these 4 different asset classes have different return, volatility and drawdown characteristics, which argues for combination and diversification. This can clearly be seen from Figure \ref{fig: data overview}.

As an investor, we hold a portfolio $\mathbf{w}$, $w_i, i \in {1, ..., U}$, which allocates a weight  $w_i$ to each investible asset. In these experiments, we attach $P \in \mathbb{N}$ sets of weights to these assets and pick a $\tau < T$ such that we have $T-\tau$ overlapping sample paths or a simple buy-and-hold strategy over $\tau$ days for every $p \in [0,...,P]$. Here we pick $\tau = 20$, so we model monthly sample paths of these mixed asset class portfolios.
    
    % \begin{figure*}
    %      \centering
    %      \begin{subfigure}{\textwidth}
    %          \centering
    %          \includegraphics[width=0.6\textwidth]{sample_paths.png}
    %      \end{subfigure}
    %      \hfill
    %         \caption{Sample paths}
    %         \label{fig: sample paths}
    % \end{figure*}

The drawdown distribution of the sample paths of such a hypothetical portfolio is shown in Figure \ref{Drawdown distribution}. If the investor finds scenarios that breach certain levels of drawdowns over the course of a particular month, she needs to reallocate and increase her weights in lower-drawdown instruments. For instance, through drawdown optimization \cite{chekhlov2005drawdown}), rescaling all portfolio weights with a risk-free cash-like fraction $w_i(1-w_{cash})$ that exhibits no drawdowns (CPPI \cite{black1992theory} or TIPP  \cite{estep1988tipp} strategy), or by buying drawdown insurance instruments such as a barrier option \cite{carr2011maximum}\cite{atteson2022carr}. However, traditional (G)BM-implied distributions understate these probabilities. In Figure \ref{Drawdown distribution}, we compare the distribution with the distribution that is implied by standard simulated Brownian motion with the same volatility $\sigma$ parameter as the portfolios we have constructed. Note that one could also use the closed forms from \citet{douady2000probability}, \citet{rej2018you} and alike here. Importantly, we notice that the blue density is merely an optimistic lower bound on the true distribution of portfolio drawdown, and misses out on the tails. This can also be seen in Figure \ref{Tail drawdown distribution}. The direct analytical link between generalized stochastic processes and this distribution is far from trivial, as one can not rely on Lévy's theorem anymore (in contrast to the theoretical (G)BM case).  In the next section, our methodology thus advocates a non-parametric approach to simulate paths that closely resemble the correct drawdown distribution. We will adapt parameterized paths $S_{\theta}$ such that their drawdown distribution converges towards this empirical drawdown distribution.

In this simulation we will:
\begin{itemize}
    \item Generate a set of portfolio weights $\mathbf{w}_p$ (index $p$ for portfolio) and construct $T-\tau$ paths, based on taking blocks $b$ from the product of $w_p$ and $S$'s cumulative returns, the univariate portfolio paths $S_p$.
    \item Calculate the drawdown $\Xi(S_p^b)$ and signatures $\Phi^M(S_p^b)$ for each block and regress the two (for increasing $M$ in [1, ..., 10]), with defining the $(1-p_{test})(T-\tau)$ first blocks as train sample and remaining blocks as test sample. Note that there is a strict train-test separation in time.
    \item Repeat $P = 100$ times and report average performance.
\end{itemize}

\paragraph{\textbf{Results and discussion}}
Our conclusions are analogous to the bottom-up simulation experiments. The accuracies are shown in Table \ref{elnet table 2}. We find consistent train and test $RMSE$ performance, which is expected given the large $K$. The average discrepancy in train-test RMSE is $0.00016041$, which is in line with Figure \ref{is oos discrepancy}.  In line with Figure \ref{Train test RMSE order}, from $M>5$ we get accurate approximations of $\xi$, after which it further improves at a much slower rate. 

The average and standard deviation of the computation time of one signature of a certain high order $M$ taken over $1000$ iterations are displayed in Figure \ref{Average compute time}. One finds diminishing improvements in $RMSE$ for exponential increase in compute time, which might be prohibitive for very large $K$. This will depend on the used hardware (the \textit{signatory} package has built-in parallelization in its C++ backend) and own parallelization choices. For this sample the improvements in accuracy become negligible after $M>5$, so it is not worth it to include higher orders of $M$. 

This can also be seen from Figures \ref{Train fit empirical data} and \ref{Test fit empirical data} which show the train and test fit respectively for one set of portfolio weights (equal weighted). It becomes apparent that signatures linearize the relationship between portfolio paths and their drawdowns well and that this applies to both smaller and higher (tail) drawdowns. From these values we can deduce that improving accuracy in the order of a few basis points (bps) does not justify the exponential increase in compute. Moreover, from a market generator application perspective (next section) the replication of a drawdown up to a few bps in the objective function is likely to be spurious precision compared to total objective value.

\begin{figure}[!h]
     \centering
         \includegraphics[width=0.40\textwidth]{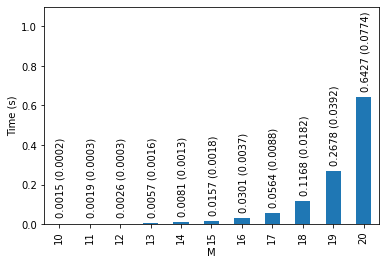}
     \hfill
        \caption{Average and standard deviation of compute times (in seconds) of a signature of order $M$ taken over 1000 iterations}
        \label{Average compute time}
\end{figure}

\begin{table}[]
\tiny
\begin{center}
\begin{tabular}{lllll}
\textbf{M (K=8429, $P$ = 100) }& \textbf{RMSE Train}        &       & \textbf{RMSE Test}        &        \\
\hline
1                   & 0.010393  & (0.002399) & 0.013613 & (0.003181) \\
2                   & 0.009492  & (0.002154) & 0.012383 & (0.002879) \\
3                   & 0.008024  & (0.001817) & 0.009181 & (0.002153) \\
4                   & 0.007005  & (0.001511) & 0.006439 & (0.001468) \\
5                   & 0.006858  & (0.001501) & 0.006240 & (0.001482)  \\
6                   & 0.006723  & (0.001387) & 0.006479 & (0.003023) \\
7                   & 0.006631  & (0.001299) & 0.006611 & (0.003950)  \\
8                   & 0.006526  & (0.001282) & 0.006659 & (0.006091) \\
9                   & 0.006491  & (0.001327) & 0.006549 & (0.006716)  \\
10                  & 0.006451  & (0.001365) & 0.006744 & (0.004720) \\
\hline
\end{tabular}
\end{center}
\caption{ElasticNet CV(10) fit for empirical data}
\label{elnet table 2}
\normalsize
\end{table}

\subsection{Drawdown market generator: results and discussion}

\begin{figure*}[h!]
     \centering
     \begin{subfigure}[b]{0.48\textwidth}
         \centering
         \includegraphics[width=\textwidth]{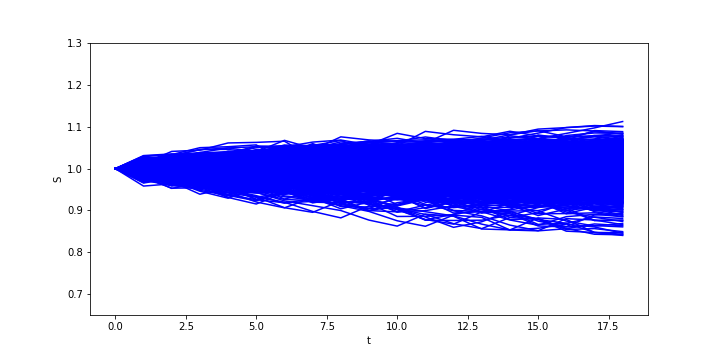}
         \caption{Generated VAE}
         \label{fig:y equals x1}
     \end{subfigure}
     \hfill
          \begin{subfigure}[b]{0.48\textwidth}
         \centering
         \includegraphics[width=\textwidth]{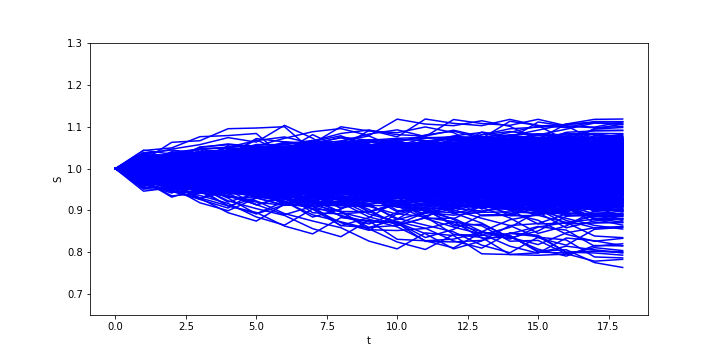}
         \caption{Generated $\xi$-VAE}
         \label{fig:y equals x2}
     \end{subfigure}
     \hfill
     \begin{subfigure}[b]{0.48\textwidth}
         \centering
         \includegraphics[width=\textwidth]{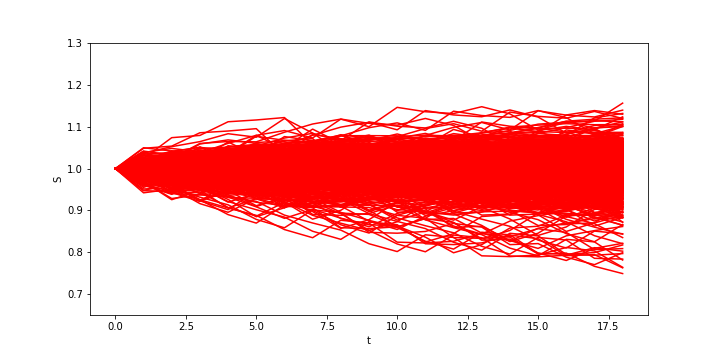}
         \caption{Original}
         \label{fig:three sin x1}
     \end{subfigure}
     \hfill
        \caption{Generated paths versus original sample paths}
        \label{fig: generated original paths}
\end{figure*}

\begin{figure*}[h!]
     \centering
     \begin{subfigure}[b]{0.48\textwidth}
         \centering
         \includegraphics[width=\textwidth]{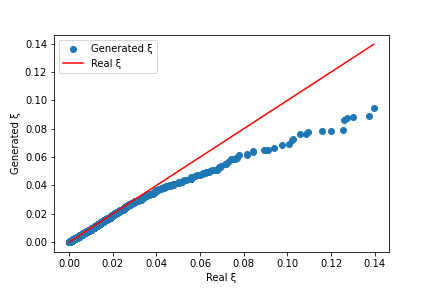}
         \caption{VAE generated versus actual $\xi$ quantiles}
         \label{fig:y equals x3}
     \end{subfigure}
     \hfill
     \begin{subfigure}[b]{0.48\textwidth}
         \centering
         \includegraphics[width=\textwidth]{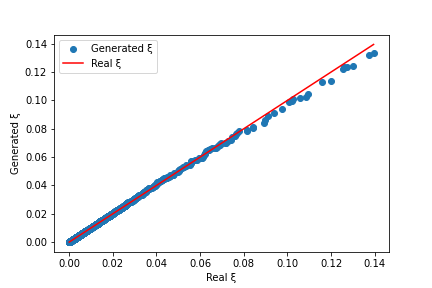}
         \caption{$\xi$-VAE generated versus actual $\xi$ quantiles}
         \label{fig:three sin x2}
     \end{subfigure}
     \hfill
        \caption{Drawdown QQ plots}
        \label{fig: qq plots}
\end{figure*}

The main findings are summarized in Figures \ref{fig: generated original paths} and \ref{fig: qq plots}. More details can again be found in the appendix on numerical results (\ref{Appendix 3: numerical results}). Figure \ref{fig: generated original paths} shows the generated paths of a standard VAE model versus the drawdown market generator, denoted by $\xi$-VAE. It is clear that the original VAE generates reasonably realistic scenarios, but is reminiscent of Figure \ref{Drawdown distribution} in the sense that the paths are non-Brownian but still too centered around the Brownian-like distribution (the densely colored areas). This is also clear from Figure \ref{fig: drawdown ret dist}. To summarize our architecture from this perspective, the standard VAE reconstruction term focuses on reproducing the top distribution ($L^2$ loss over a particular batch) such that it misses out on the tails. The $\xi$-VAE fits both the top and the bottom distribution such that it includes more tail (drawdown) scenarios. In sum, we find that the standard VAE shows a lack of more extreme adverse scenarios. 

This also becomes apparent in Figure \ref{fig: drawdown scatter}, which plots the synthetic and actual drawdowns as a scatter. The standard VAE does a good job in capturing most part of the moderate drawdowns, but is scattered towards the higher drawdowns and produces none in the tail, while this is resolved with the $\xi$-VAE. Moreover, from the cumulative drawdown distribution the Kolmogorov-Smirnoff test that the synthetic drawdowns come from the same distribution as the original one cannot be rejected for both the standard and $\xi$-VAE, as we know the standard KS statistic is not sensitive to diverging tail distributions. A closely related graph is Figure \ref{fig: qq plots}, which plots the same scatter but using ordered observations (i.e. quantiles) and a first bisector line which would fit the data if the synthetic drawdown distribution has exactly the same quantiles as the original distribution. Again we find that the VAE misses out on capturing tail drawdown scenarios, while $\xi$-VAE does a much better job. This is not done by just memorizing the training samples and reproducing them exactly as the bottleneck VAE architecture forces a lower-dimensional representation of the training data (i.e. the non-parametric DGP). The autoencoder serves as a dimension reduction mechanism, rather than trivially mirroring the training data. That is why one sees different simulated paths than identical ones to the input samples, and why by using the trained decoder as a non-parametric DGP one can create an infinite amount of genuinely new data with the same distribution as the training data. Moreover, both train and validation losses flat out at convergence (please find Appendix \ref{Appendix 3: drawdown market generator details} with details on the training and test convergence) while purely mirroring data would imply training losses would be minimal at the cost of high validation losses. 

\section{Conclusion}\label{conclusion}
Learning functions on paths is a highly non-trivial task. Rough path theory offers tools to approximate functions on paths that have sufficient regularity by considering a controlled differential system and iterating the effect of intervals of the (rough) path on the (smooth) outcome function. One key path dependent risk functional in finance is a portfolio value path's drawdown. This paper takes the perspective of portfolio drawdown as a non-linear dynamic system, a controlled differential equation, rather than directly analyzing the solution, or the exact expression of drawdown containing the running maximum operator. It relied on some important insights from the theory of rough paths to pinpoint that by taking this perspective, rather than continuous differentiability, a more general regularity condition of bounded changes in drawdown as a function of changes in input paths is sufficient to use a non-commutative exponential to interpolate between the two types of path dependent effects of a driving path on its resulting drawdown by numerically approximating the average nested effect. This thus allows one to locally approximate the drawdown function without having to evaluate its exact expression and in the meantime leapfrogging the inherent path dependence of its time derivatives. The linear dependence of a path's drawdown on its differentiable signature representation then allows one to embed drawdown evaluations into systems of differentiable equations such as generative ML models. We prove this required regularity w.r.t. drawdown: the boundedness of its time derivatives allows us to write it in a more general controlled differential equation notation, and the Lipschitz regularity of the drawdown function assures bounded errors in its convergence proofs. That is why we then illustrated the consistency of the approximation: on simulated fractional Brownian and on real-world data, regression results exhibit a good fit for penalized linear regression (Elastic Net regression) when one has a reasonable sample size. Finally, our proposed application of the approximation is a so-called market generator model that evaluates the synthetic time series samples in terms of their drawdown. We argue that by including a drawdown learning objective, upon convergence of this reconstruction term in train and validation steps, one gets more realistic scenarios than the standard VAE model that exhibit quantifiably (i.e. the measured loss convergence) close drawdowns to the empirical ones, hence effectively reproducing the drawdown distributions without trivially mapping input paths to output paths.

Future work will focus on extending this application and further applying it to portfolio drawdown optimization, where one can for example test a data-hungry drawdown control strategy over a host of synthetic scenarios rather than the single historical one. In this context, learning drawdown scenarios can be seen as a denoising mechanism to first remove noise (encoding), then adding new noise (decoding new samples), then constructing an ensemble average strategy over a host of noisy scenarios that cancels out by construction instead of by assumption for historical scenarios (e.g. bootstrap methods). This non-parametric Monte Carlo idea could hence robustify one's methodology further as a mathematically principled data augmentation technique.  Moreover, the non-parametric nature of our Monte Carlo engine opens possibilities to full non-parametric pricing of path dependent (e.g. barrier) max drawdown insurance contingent claims.

\clearpage
\appendix
\section*{Appendix: Controlled differential equations and path signatures}\label{appendix}
\subsection*{Controlled differential equations}
    We are generally interested in a CDE of the form:
    \begin{equation}
        dY_t = g(Y_t)dX_t
    \end{equation}
    where $X$ is a continuous path on $[0,T] \xrightarrow{} \mathbb{R}$, called the driving signal of the dynamic system. $g$ is a $\mathbb{R} \xrightarrow{} \mathbb{R}$ mapping called the physics that models the effect of $dX_t$ on the response $dY_t$. A controlled differential equation (CDE) distinguishes itself from an ordinary differential equation in the sense that the system is controlled or driven by a path ($dX$) rather than time ($dt$) or a random variable (stochastic SDEs, $d\varepsilon$).

\subsection*{Signatures}
    A series of coefficients of the path that naturally arrives from this type of equation is the series of iterated integrals of the path, or the path signature $\Phi$. The signature of a path $X: [0, T] \xrightarrow{} \mathbb{R}$ can be defined as the sequence of ordered coefficients:
    \begin{equation}
        \Phi(X) = (1, \Phi_1, ..., \Phi_n,...)
    \end{equation}
    where for every integer n (order of the signature):
    \begin{equation}\label{eq sig def}
        \Phi_n(X) = \underset{u_1<...<u_n, u_1,...,u_n \in [0,T]}{\int...\int}dX_{u_1}\otimes...\otimes dX_{u_n}
    \end{equation}
    where we define the $n$-fold iterated integral as all the integrals over the $n$ ordered intervals $u_i$ in [0,T]. The signature is the infinite collection for $n \xrightarrow{} \infty$, although typically lower level M truncations are used. 
    \begin{equation}
        \Phi^M(X) = (1, \Phi_1, ..., \Phi_M)
    \end{equation}
        
\subsection*{Picard Iterations}

    The idea behind a Picard iteration is to define for:
    \begin{equation}
        dY_t = g(Y_t)dX_t
    \end{equation}
    a sequence of mapping functions $Y(n): [0,T]  \xrightarrow{} \mathbb{R}$ recursively such that for every $t \in [0,T]$:
    \begin{equation}
        Y(0)_t \equiv y_0
    \end{equation}
    \begin{equation}
        Y(1)_t = y_0 + \int_{0}^{t}g(y_0)dX_s
    \end{equation}
    \begin{equation}
        Y(n+1)_t = y_0 + \int_{0}^{t}g(Y(n)_s)dX_s
    \end{equation}     

    Now by simple recursion one finds that (for a linear $g$):
    \begin{equation}
        Y(n)_t = y_0 + \sum_{k}^{n}g^{\otimes k}(y_0)\underset{u_1<...<u_n, u_1,...,u_n \in [0,T]}{\int...\int}dX_{u_1}\otimes...\otimes dX_{u_n}
    \end{equation}
    Such that a solution for $Y_t$ would be:
    \begin{equation}\label{eq yt approx}
        Y_t = y_0 + \sum_{k}^{\infty}g^{\otimes k}(y_0)\underset{u_1<...<u_k, u_1,...,u_k \in [0,T]}{\int...\int}dX_{u_1}\otimes...\otimes dX_{u_k}
    \end{equation}
    This result shows how the signature, as an iterative representation of a path over ordered intervals, naturally arises from solving CDEs using Picard iterations, and how it is a natural generalization of Taylor series on the path space when the physics is linear.   

    \begin{equation}
        g^{\circ 1} = g
    \end{equation}  
    \begin{equation}
        g^{\circ n+1} = D(g^{\circ n})g
    \end{equation}
    then it is natural to define the $N$-step Taylor expansion for $Y_t$ by $\hat{Y}(N)_t$ as:
    \begin{equation}
        \hat{Y}(N)_t = y_0 + \sum_{n=1}^{N}g^{\circ n}(y_0)\underset{u_1<...<u_n, u_1,...,u_n \in [0,T]}{\int...\int}dX_{u_1}\otimes...\otimes dX_{u_n}
    \end{equation}
    Clearly, $\hat{Y}(N)_t$ is linear in the truncated signature of X up to order N\footnote{\label{footnote factorial}Moreover, the error bounds of $\hat{Y}(N)_t$ to approximate $Y_t$ yield a factorial decay in terms of N, i.e. $|Y_t - \hat{Y}(N)_t| \leq C\frac{|X|^{N+1}_{1, [0,t]}}{N!}$. This result can be extended to p-geometric rough paths where g is a $Lip(K)$ where $K>p-1$ \cite{boedihardjo2015uniform}.}.  
    
\paragraph{\textbf{Example}} The simplest example of:
\begin{equation}
    dY_t = g(Y_t)dX_t
\end{equation}
is a linear physics for a linear path X:
\begin{equation}
    dY_t = Y_tdX_t
\end{equation}
where:
\begin{equation}
    g = g^{\circ 1}
\end{equation}
\begin{equation}
    g^{\circ n + 1} = D(g^{\circ n})g
\end{equation}
\begin{equation}
    X_t = X_0 + \frac{X_T - X_0}{T} t
\end{equation}
and assuming:
\begin{equation}
    y_0 = 1
\end{equation}
\begin{equation}
    X_0 = 0
\end{equation}

Indeed, this yields the exponential function $Y_t = \exp(X_t)$. For non-linear driving signals (where the order of the events matter), one generally gets a non-commutative version of the exponential function in Eq. (\ref{eq yt approx})! For linear time, the order of events does not matter and we generally get the increment of the path raised to the level of the iterated integral, divided by the level factorial (i.e. the area of an $n$-dimensional simplex).

This can be seen from:
\begin{equation}\label{eq exp approx}
    Y_t = y_0 + \sum_{n=1}^{N}Y^{\circ n}\underset{u_1<...<u_n, u_1,...,u_n \in [0,T]}{\int...\int}dX_{u_1}\otimes...\otimes dX_{u_n}
\end{equation}
now it is easy to see that:
\begin{equation}
\begin{split}
    \Phi^n = \underset{u_1<...<u_n, u_1,...,u_n \in [0,t]}{\int...\int}dX_{u_1}\otimes...\otimes dX_{u_n} \\ = \underset{0 < u_1 < ... < u_n < t}{\int...\int}d(\frac{X_t}{t}u_1)...d(\frac{X_t}{t}u_n) \\ = \prod_{j=1}^{n}(\frac{X_t}{t})\underset{0 < u_1 < ... < u_n < t}{\int...\int}du_1...du_n \\ = \frac{1}{t^N}\prod_{j=1}^{n}(X_t)\frac{t^N}{n!} = \frac{(X_t)^n}{n!}
\end{split}
\end{equation}
such that:
\begin{equation}
    Y_t = y_0 + \sum_{n=1}^{N}y_0\frac{1}{n!}(X_t)^n = 1 + (X_t) + \frac{(X_t)^2}{2!} + \frac{(X_t)^3}{3!} + ...
\end{equation}
Which is the classical Taylor expansion for the exponential function, i.e. linear physics integrated over a linear path (time). Now the signature approximation is the generalization of this idea to the path space (i.e. Y is a function on a path), where the path $X_t$ need not be a linear map of time, and the physics $g$ need not be linear (e.g. drawdown Eq. (\ref{f definition})).

\section*{Appendix 2: Variational Autoencoders}\label{Appendix 2: variational autoencoder}

This section gives an overview of the variational autoencoder architecture. VAEs converge fast\footnote{First experiments with VAE resulted in similar performance metrics with GAN, where VAE was trained c.30 seconds and GAN c.30 minutes.}, are more stable than competing architectures, and they allow us to interpret the (conditional) distributions after training. 

Below, we discuss the mappings $f_{\Theta}(X)$ and $f^{-1}_{\Theta}(Z)$ between the data and latent distributions, the original loss function $\mathcal{L}(X, X')$ (which we revised in drawdown terms in this paper), the training algorithm and the hyperparameters.

\subsection*{General architecture}

\begin{figure}[!h]
    \hspace*{-0.75cm}   
         \includegraphics[width=0.60\textwidth]{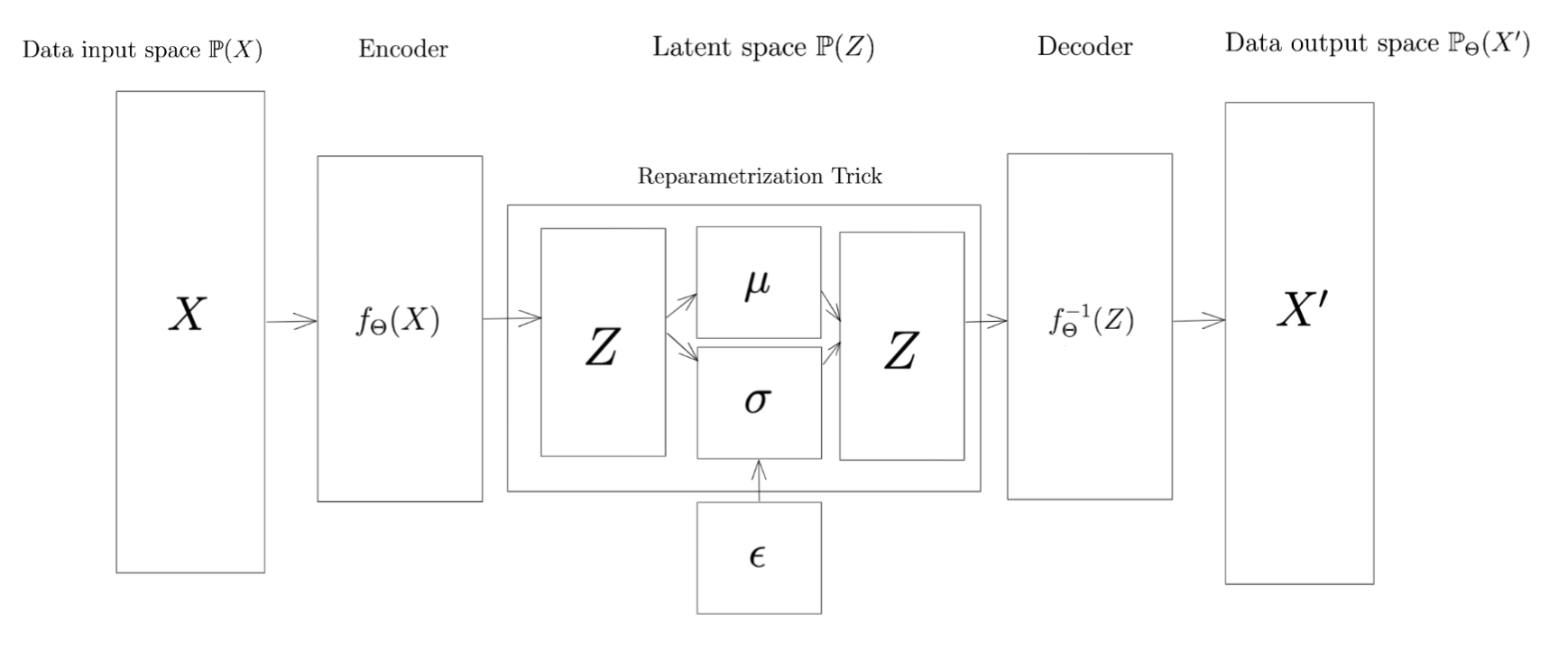}
     \hfill
        \caption{VAE architecture overview}
        \label{VAE architecture overview}
\end{figure}

The architecture of a VAE is summarized in Figure \ref{VAE architecture overview}. As input we have a $D_p$-dimensional space $X$, the physical data domain that we can measure. Using a flexible neural network mapping $f_{\Theta}: \mathbb{R}^{D_p} \rightarrow \mathbb{R}^{D_l}, D_l << D_p$, called the encoder, we compress the dimension of the data into a $D_l$-dimensional latent space $Z$, e.g. 10-dimensional. Using the reparametrization trick (\cite{kingma2014semi}) we map $Z$ onto a mean $\mu$ and standard deviation $\sigma$ vector, i.e. onto a $D_l$-dimensional Gaussian, e.g. a 10-dimensional multi-variate normal distribution. 
Starting from multi-variate normal data, we can recombine $\mu$ and $\sigma$ into a $D_l$-dimensional $Z$. The decoder neural network $f^{-1}_{\Theta}: \mathbb{R}^{D_l} \rightarrow \mathbb{R}^{D_p}$ maps the latent space back to the output space $\mathbb{P}_{\Theta}(X')$ where $X'$ can be considered reconstructed samples in the training step, or genuinely new or synthetic samples in a generator step. The quality of the VAE clearly depends on the similarity between $\mathbb{P}(X)$ and $\mathbb{P}_{\Theta}(X')$.

\subsection*{Encoder - decoder networks}
Let us now zoom in on $f_{\Theta}(X)$ and $f^{-1}_{\Theta}(Z)$. Each neural network consist of one layer of $J$ mathematical units called neurons:

\begin{equation}
    f_{\Theta_j} = A(\sum_i^{D_l} \theta_{i,j} x_i)
\end{equation}
 
 Every neuron takes linear combinations $\theta_i$ of the input data point $x_i$ and is then activated using a non-linear activation function $A$, such as rectified linear units (ReLU), hyperbolic tangent (tanh) or sigmoid. In this paper we use a variant of ReLU called a leaky ReLU:
 
\begin{equation}
    LReLU(x) = 1_{x < 0} \alpha x + 1_{x \geq 0} x
\end{equation}
where $\alpha$ is a small constant called the slope of the ReLU. 
The $J$ neurons are linearly combined into the next layer (in this case $Z$):
\begin{equation}
    Z_k := \sum_j^J \theta_{j,k} f_{\Theta_j}
\end{equation}
for every $k$ in $D_l$. The decoder map can formally be written down like the encoder, but in reverse order.
\subsection*{Loss function}
The loss function of a VAE generally consists of two components, the latent loss ($\mathcal{L}_L$) and the reconstruction loss ($\mathcal{L}_R$):
\begin{equation}
    \mathcal{L}(X, X') = \beta  \mathcal{L}_L + (1-\beta)  \mathcal{L}_R
\end{equation}
The latent loss is the Kullback-Leibler discrepancy between the latent distribution under its encoded parametrization, the posterior $f_{\Theta}(X) = \mathbb{P}_{\Theta}(Z|X)$, and its theoretical distribution, e.g. multi-variate Gaussian $\mathbb{P}(Z)$. Appendix B in \cite{kingma2014stochastic} offers a simple expression for $\mathcal{L}_L$.
The reconstruction loss is the cost of reproducing $\mathbb{P}_{\Theta}(X')$ after the dimension reduction step, and originally computed by the root of the mean squared error (RMSE or $L2$-loss) between X and X'.
\begin{equation}
\begin{split}
    \mathcal{L}(X, X') = \beta  \frac{1}{2}\sum_k^K(1 + \sigma - \mu^2 - \exp(\sigma)) \\ + (1-\beta)  \mathbb{E}(||X-X'||^2)
\end{split}
\end{equation}

\subsection*{Training}
In terms of training, the learning algorithm is analoguous to most deep learning methods. Optimal loss values $\mathcal{L}^*$ are determined by stochastically sampling batches of data and alternating forward and backward passes through the VAE. For each batch the data is first passed through the encoder network and decoder network (forward pass), after which $\mathcal{L}$ is evaluated in terms of $\Theta$. At each layer, the derivative of $\mathcal{L}$ vis-a-vis $\Theta$ can easily be evaluated. Next (backward pass), we say the calculated loss backpropagates through the network, and $\Theta$ are adjusted in the direction of the gradient $\nabla_{\Theta}\mathcal{L}$ with the learning rate as step size. The exact optimizer algorithm we used for this is Adam (Adaptive moments estimation, \cite{kingma2014adam}). Finally, we also use a concept called regularization, which penalizes neural models that become too complex or overparametrized. We used a tool called dropout, that during training randomly sets a proportion of parameters in $\Theta$ equal to zero, and leaves those connections at zero that contribute the least to the prediction.
\subsection*{Hyperparameters}
In summary, the hyperparameters of this architecture are: (1) the number of neurons in the encoder, (2) the number of neurons in the decoder, (3) the latent dimension $D_l$, (4) the learning rate $l$, (5) the optimizer algorithm, (6) the dropout rate, (7) the batch size $N_b$, (8) batch length $\tau$ and (9) number of training steps $N_t$. We opted for the following set up, which was optimized using grid search: 50, 50, 10, 0.001, Adam, 0.01, 50, 20, 200 (with early stopping criteria\footnote{Not all $N_t$ = 200 steps are executed if the objective values (both total and individual terms) are not improved over e.g. the last $I$ training iterations. For a standard VAE this is the total objective value, latent loss and the reconstruction loss, for the $\xi$-VAE drawdown convergence is now an additional criterion. After some fine-tuning $I$ was set to $3$.}). 

\subsection*{Generation}
After training, in the sampling or generation step, we start from a random $D_l$-dimensional noise $\epsilon \sim \mathbb{P}(Z)$ which is $D_l$-variate Gaussian. Now, we simply need one decode step to generate new samples of $\mathbb{P}_{\Theta}(X')$.

\clearpage
\newpage
\section*{Appendix 3: Numerical results}\label{Appendix 3: numerical results}

\parbox{15cm}{\subsection*{Consistency and convergence of linear drawdown approximation in the signature space}}
\begin{figure}[h!]
     \centering
     \begin{subfigure}{\textwidth}
         \centering
         \includegraphics[width=0.4\textwidth]{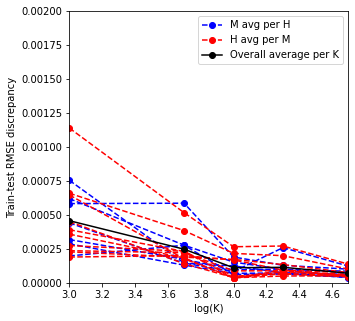}
     \end{subfigure}
    \parbox{15cm}{\caption{\textit{In- and out-of-sample accuracy difference as a function of sample size $log(K)$ (Red = averaging out over $H$ to get a line per value of $M$, blue = opposite and black is overall average over both M and H).}}\label{is oos discrepancy}}
     \hfill
    \quad
\end{figure}

% \begin{figure*}
%      \centering
%      \begin{subfigure}{\textwidth}
%          \centering
%          \includegraphics[width=0.75\textwidth]{MERGED order R2 with detail Updated June.png}
%      \end{subfigure}
%      \hfill
%     \quad
%      \hfill
%         \caption{Train and test $R^2$ as a function of signature approximation order $M$. Average over $K$ and $H$ is shown on the left. An average per $K$ is shown on the right.}
%         \label{Train test R2 order}
% \end{figure*}

% \clearpage
% \newpage
\begin{figure}[h!]
     \centering
     \begin{subfigure}{\textwidth}
         \centering
         \includegraphics[width=0.75\textwidth]{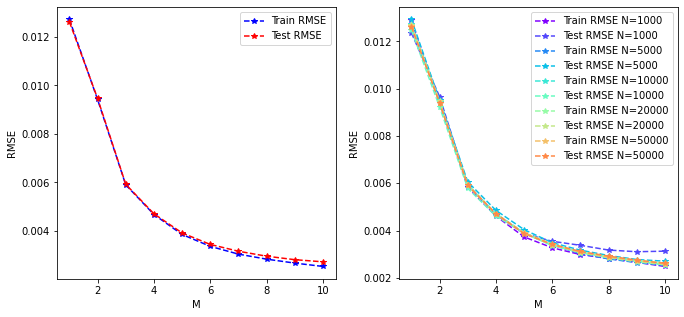}
     \end{subfigure}
     \hfill
    \quad
     \hfill
        \parbox{15cm}{\caption{\textit{Train and test $RMSE$ (y-axis) as a function of signature approximation order $M$. Average over $K$ and $H$ is shown on the left. An average per $K$ is shown on the right.}}\label{Train test RMSE order}}
\end{figure}

\clearpage
\begin{figure}
     \centering
     \begin{subfigure}{\textwidth}
         \centering
         \includegraphics[width=0.4\textwidth]{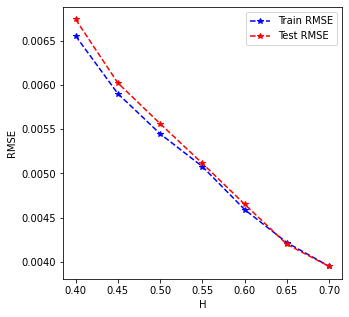}
     \end{subfigure}
     \hfill
    \quad
     \hfill
     \centering
        \caption{\textit{Train and test $RMSE$ as a function of the roughness $H$}}
        \label{Train test R2 roughness}
\end{figure}

% \begin{figure*}
%      \centering
%      \begin{subfigure}{\textwidth}
%          \centering
%          \includegraphics[width=0.75\textwidth]{Factorial decay overview R2.png}
%      \end{subfigure}
%      \hfill
%     \quad
%      \hfill
%         \caption{Factorial error (R2) decay as a function of $H$}
%         \label{Factorial error (R2) decay as a function of H}
% \end{figure*}

\subsection*{Empirical data overview}\label{Appendix 3: empirical data overview}

\begin{figure}[h!]
     \centering
     \begin{subfigure}{\textwidth}
         \centering
         \includegraphics[width=0.85\textwidth]{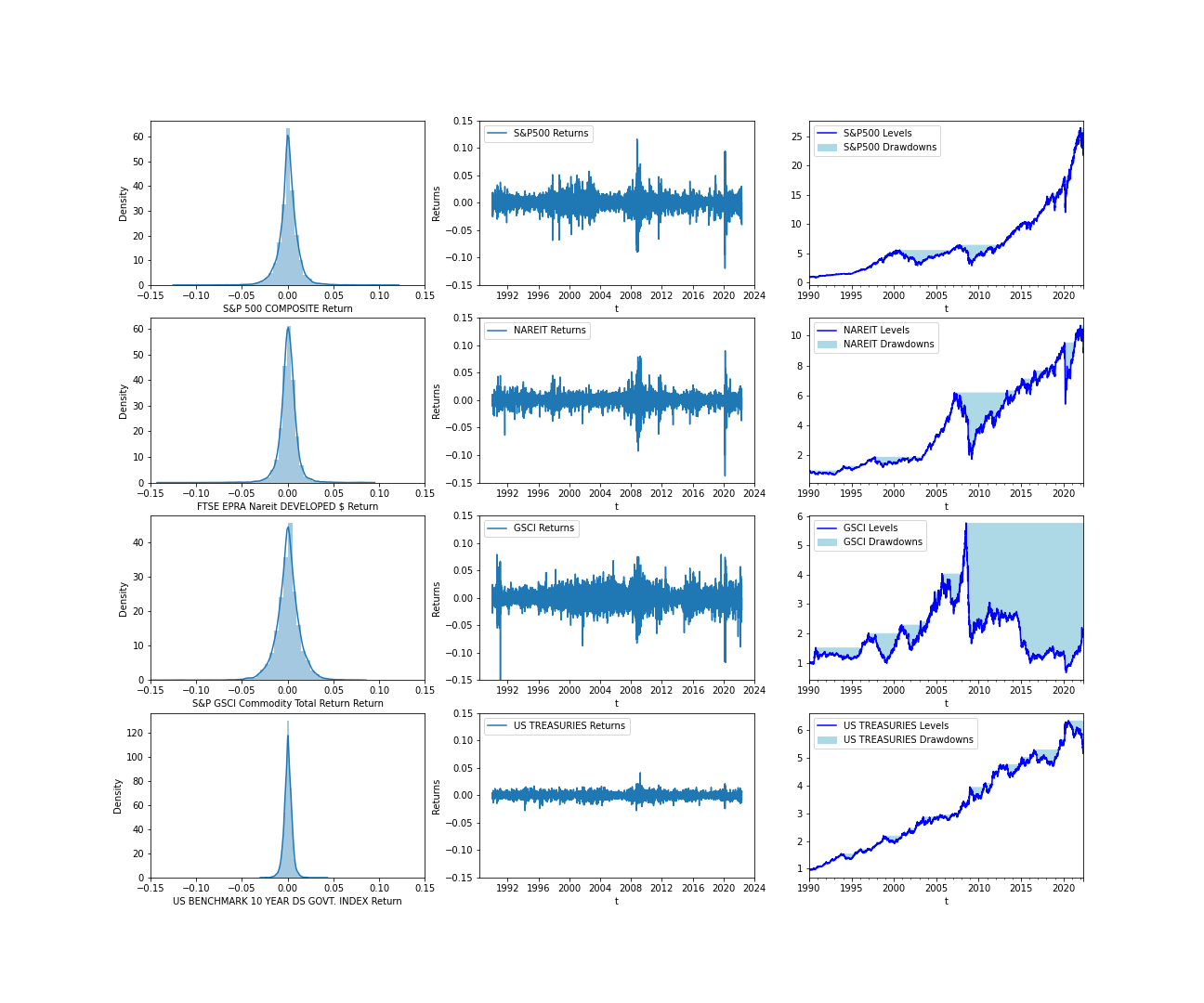}
     \end{subfigure}
     \hfill
        \parbox{11cm}{\caption{\textit{Data overview: return, volatility and drawdown heterogeneity among asset classes.}} \label{fig: data overview}}
\end{figure}

\clearpage
\subsection*{Empirical drawdown distribution}\label{Appendix 3: empirical drawdown distribution}

\begin{figure}[h!]
     \centering
     \begin{subfigure}{\textwidth}
         \centering
         \includegraphics[width=0.60\textwidth]{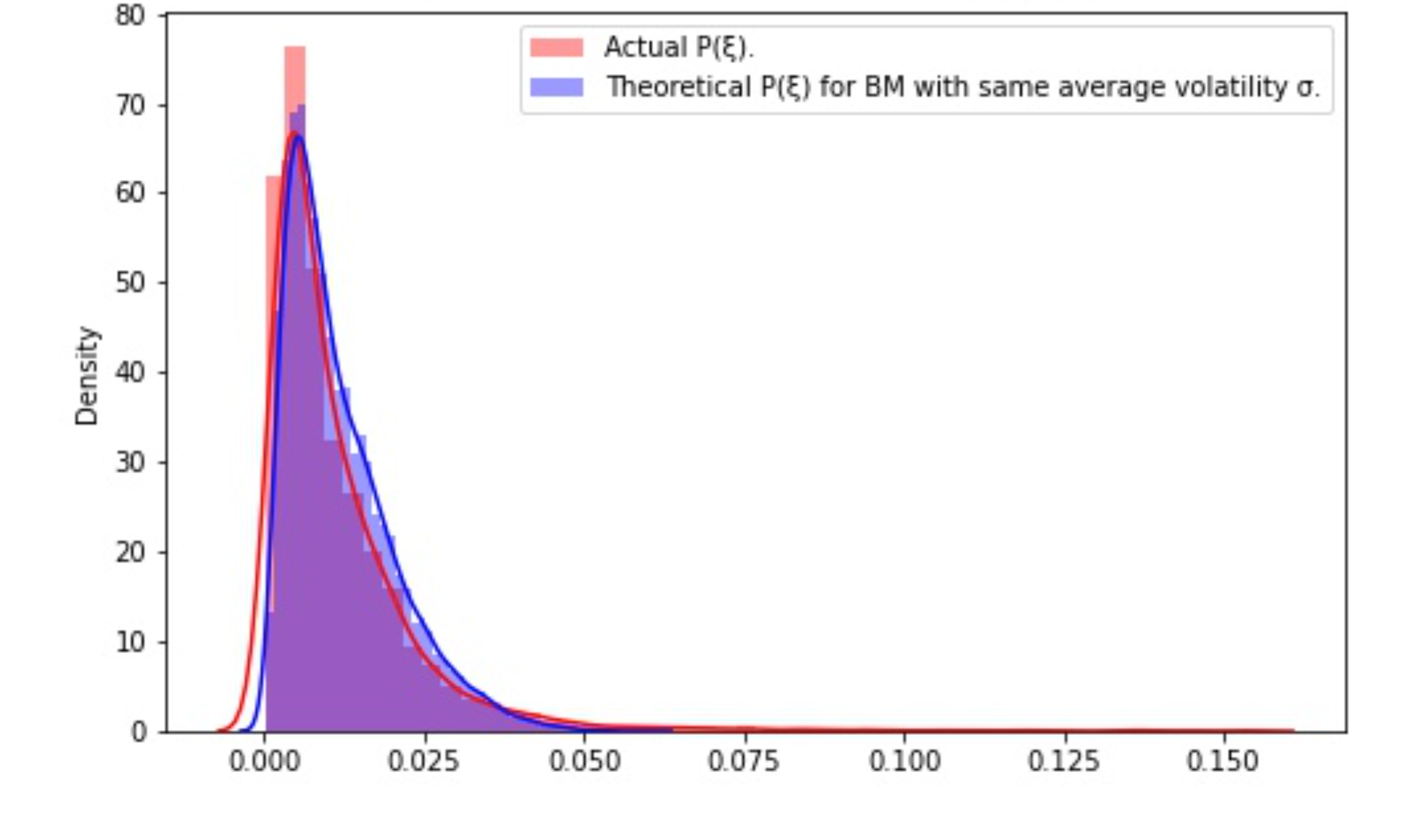}
     \end{subfigure}
     \hfill
        \parbox{15cm}{\caption{\textit{Drawdown distribution of a real-world mixed asset class portfolio versus the fitted theoretical drawdown distribution if the underlying DGP would be Brownian Motion (BM) with the same average volatility as the sample paths.}}\label{Drawdown distribution}}
\end{figure}

\begin{figure}[h!]
     \centering
     \begin{subfigure}{\textwidth}
         \centering
         \includegraphics[width=0.60\textwidth]{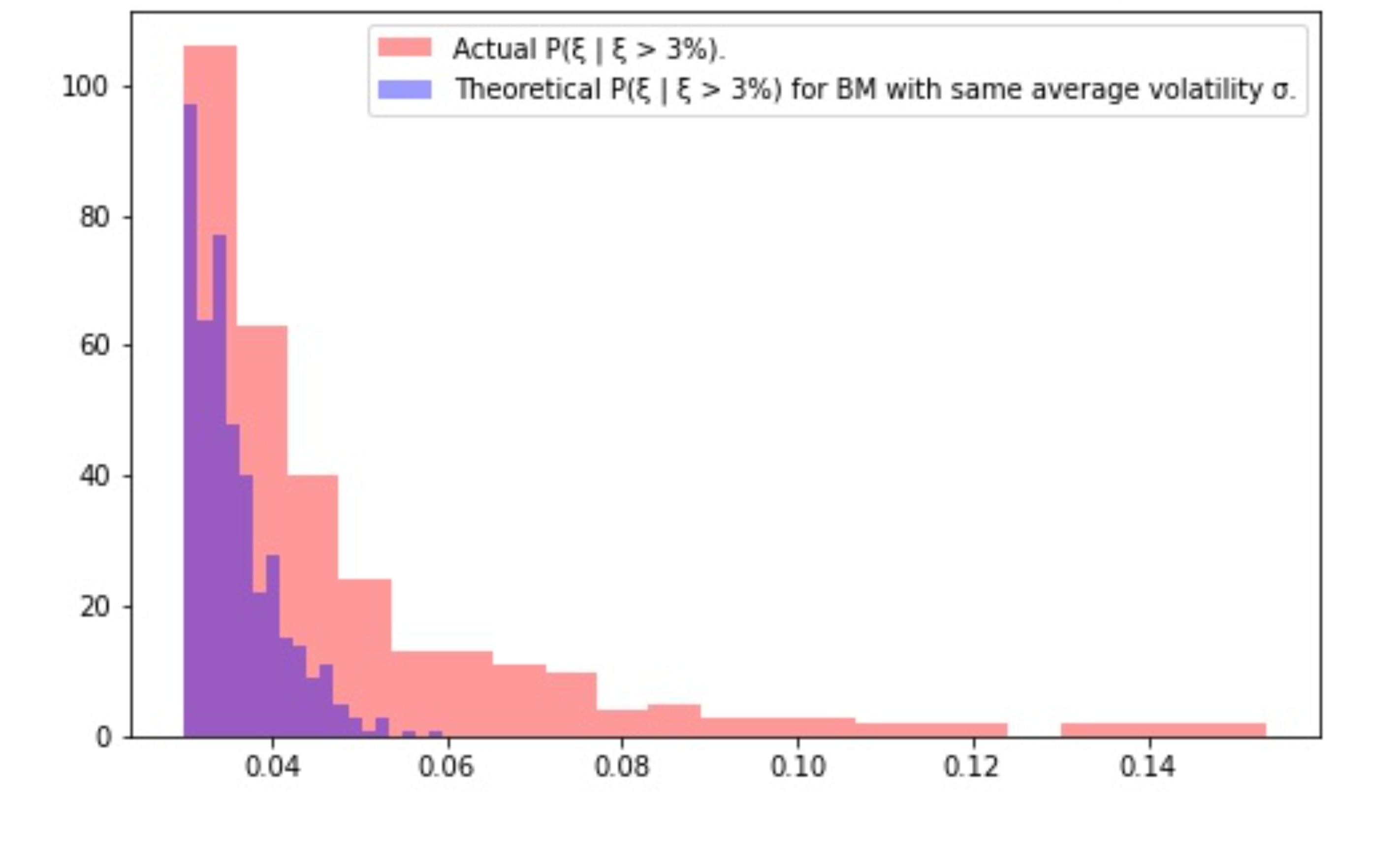}
     \end{subfigure}
     \hfill
        \parbox{10cm}{\caption{\textit{Zoom on the tail of the empirical versus theoretical drawdown distribution.}}\label{Tail drawdown distribution}}
\end{figure}

\clearpage
\subsection*{Empirical data: drawdown approximation fit}\label{Appendix 3: empirical data drawdown approximation fit}

\begin{figure}[h!]
     \centering
     \begin{subfigure}{\textwidth}
         \centering
         \includegraphics[width=0.80\textwidth]{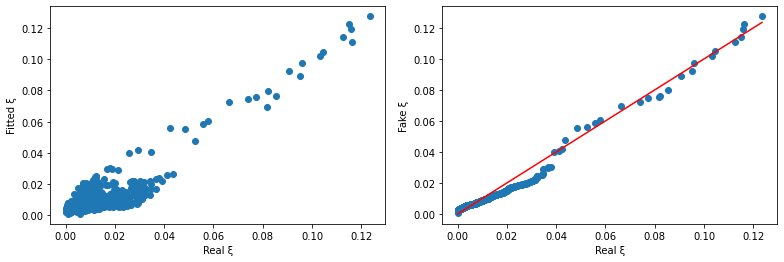}
     \end{subfigure}
     \hfill
        \caption{Train fit empirical data for one set of portfolio weights (equal weighted): scatter (left) and QQ plot (right)}
        \label{Train fit empirical data}
\end{figure}

\begin{figure}[h!]
     \centering
     \begin{subfigure}{\textwidth}
         \centering
         \includegraphics[width=0.80\textwidth]{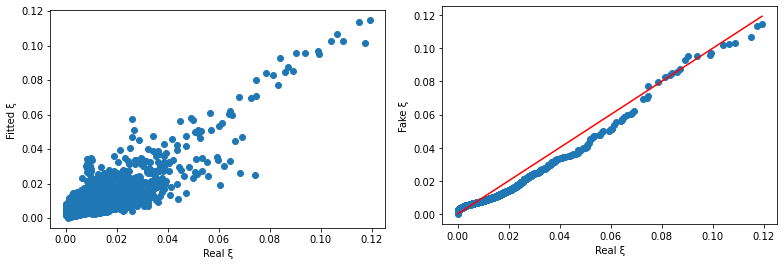}
     \end{subfigure}
     \hfill
        \caption{Test fit empirical data for one set of portfolio weights (equal weighted): scatter (left) and QQ plot (right)}
        \label{Test fit empirical data}
\end{figure}

\subsection*{Drawdown market generator}\label{Appendix 3: drawdown market generator results}

\begin{figure*}[ht!]
     \centering
          \begin{subfigure}[b]{0.48\textwidth}
         \centering
         \includegraphics[width=\textwidth]{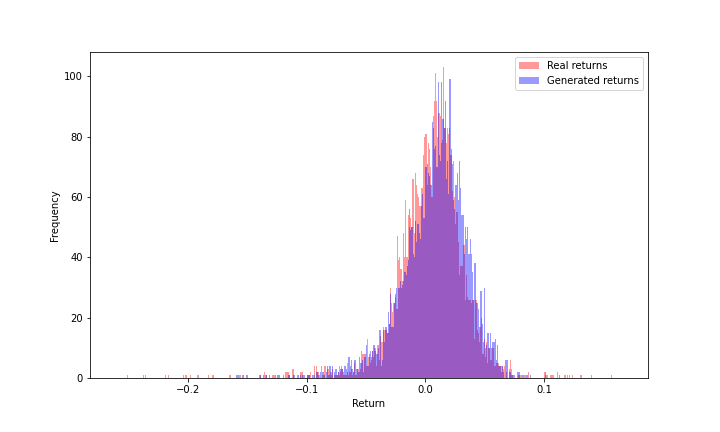}
         \caption{Return distribution VAE}
         \label{fig:y equals x4}
     \end{subfigure}
     \hfill
     \begin{subfigure}[b]{0.48\textwidth}
         \centering
         \includegraphics[width=\textwidth]{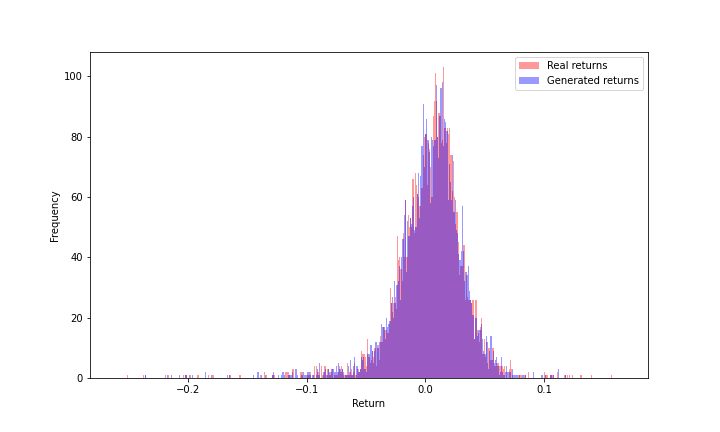}
         \caption{Return distribution $\xi$-VAE}
         \label{fig:three sin x3}
     \end{subfigure}
     \hfill
     \begin{subfigure}[b]{0.48\textwidth}
         \centering
         \includegraphics[width=\textwidth]{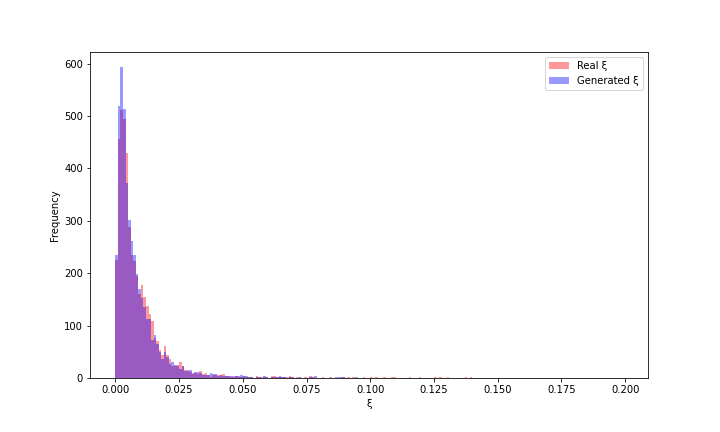}
         \caption{Drawdown distribution VAE}
         \label{fig:y equals x5}
     \end{subfigure}
     \hfill
     \begin{subfigure}[b]{0.48\textwidth}
         \centering
         \includegraphics[width=\textwidth]{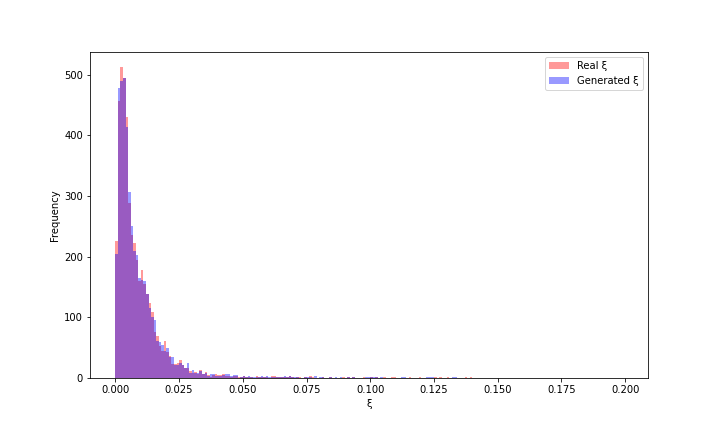}
         \caption{Drawdown distribution $\xi$-VAE}
         \label{fig:three sin x4}
     \end{subfigure}
     \hfill
        \caption{Generated versus actual drawdowns and returns distribution}
        \label{fig: drawdown ret dist}
\end{figure*}

\begin{figure*}[ht!]
     \centering
     \begin{subfigure}[b]{0.48\textwidth}
         \centering
         \includegraphics[width=\textwidth]{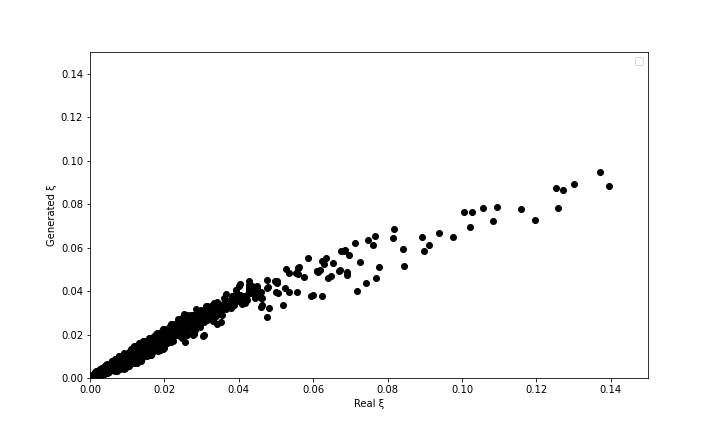}
         \caption{VAE}
         \label{fig:y equals x6}
     \end{subfigure}
     \hfill
     \begin{subfigure}[b]{0.48\textwidth}
         \centering
         \includegraphics[width=\textwidth]{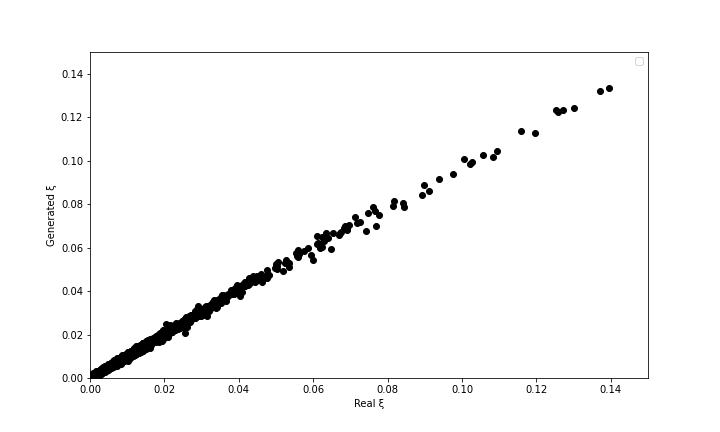}
         \caption{$\xi$-VAE}
         \label{fig:three sin x5}
     \end{subfigure}
     \hfill
        \caption{Generated versus actual drawdowns scatter }
        \label{fig: drawdown scatter}
\end{figure*}

\clearpage
\subsection*{Drawdown market generator loss convergence}\label{Appendix 3: drawdown market generator details}

\begin{figure}[!h]
     \centering
     \begin{subfigure}{\textwidth}
         \centering
         \includegraphics[width=0.6\textwidth]{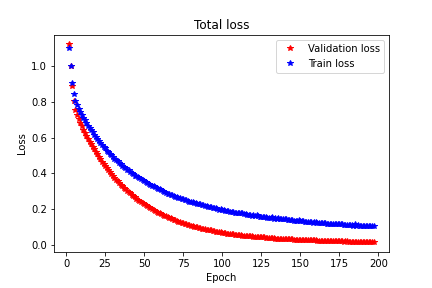}
     \end{subfigure}
     \hfill
        \parbox{9cm}{\caption{Drawdown market generator: total loss convergence (rescaled)}}
        \label{Total loss}
\end{figure}

\begin{figure}[!h]
     \centering
     \begin{subfigure}{\textwidth}
         \centering
         \includegraphics[width=0.6\textwidth]{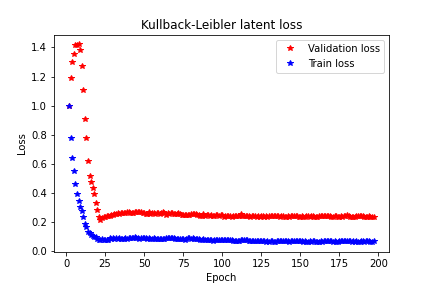}
     \end{subfigure}
     \hfill
        \parbox{9cm}{\caption{Drawdown market generator: latent (KL) loss convergence (rescaled)}}
        \label{KL loss}
\end{figure}

\clearpage
\begin{figure}[!h]
     \centering
     \begin{subfigure}{\textwidth}
         \centering
         \includegraphics[width=0.6\textwidth]{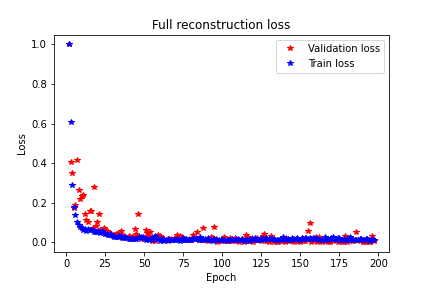}
     \end{subfigure}
     \hfill
        \parbox{10cm}{\caption{Drawdown market generator: total reconstruction loss convergence (rescaled)}}
        \label{RL loss}
\end{figure}

\begin{figure}[!h]
     \centering
     \begin{subfigure}{\textwidth}
         \centering
         \includegraphics[width=0.6\textwidth]{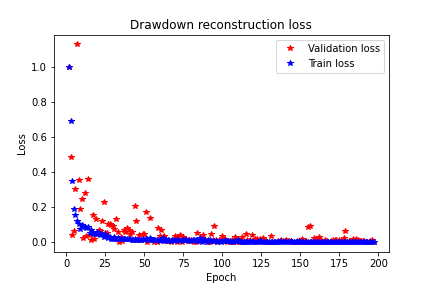}
     \end{subfigure}
     \hfill
        \parbox{11cm}{\caption{Drawdown market generator: drawdown reconstruction loss convergence (rescaled)}}
        \label{DD loss}
\end{figure}

\clearpage
\newpage
\bibliographystyle{model1-num-names}
\bibliography{main.bib}






% \end{thebibliography}

\end{document}